\documentclass[11pt]{article}

\usepackage{a4wide}
\usepackage{authblk}
\usepackage{amsmath,amssymb}
\usepackage{amsthm}
\usepackage[dvipdfmx]{graphicx}

\usepackage{ascmac}
\usepackage{tabls}
\usepackage{booktabs}
\usepackage[dvipdfmx]{graphicx}
\usepackage{wrapfig}
\usepackage{array}
\usepackage[algoruled,linesnumbered,algo2e,vlined]{algorithm2e}
\usepackage{url}
\usepackage{multirow}
\usepackage{multicol}
\usepackage{cases}
\usepackage[normalem]{ulem}
\usepackage{boites}

\usepackage[final,mode=multiuser]{fixme}

\usepackage{colortbl}

\fxsetup{theme=color}
\definecolor{fxtarget}{rgb}{0.0000,0.0000,0.4823}

\fxuseenvlayout{colorsig}
\fxusetargetlayout{color}

\fxuseenvlayout{colorsig}
\fxusetargetlayout{color}
\FXRegisterAuthor{yy}{ayy}{YY}
\FXRegisterAuthor{yn}{ayn}{YN}
\FXRegisterAuthor{si}{asi}{SI}
\fxsetface{env}{}

\pdfoutput=1

\usepackage{thmtools}

\newtheorem{theorem}{Theorem}

\theoremstyle{definition}

\newtheorem{problem}{Problem}

\renewenvironment{proof}{\begin{trivlist} \item{\textit{Proof.}}}{\end{trivlist}}

\newcommand{\Subseq}{\mathsf{Subseq}}

\newcommand{\lf}{\mathit{L}}

\renewcommand{\P}{\mathsf{P}}
\newcommand{\LMP}{\mathsf{LMP}}
\newcommand{\MP}{\mathsf{MP}}

\newcommand{\ind}{\mathsf{in\_deg}}

\begin{document}

\title{Computing SEQ-IC-LCS of Labeled Graphs}

\author[1]{Yuki~Yonemoto}
\author[2]{Yuto~Nakashima}
\author[2]{Shunsuke~Inenaga}

\affil[1]{Department of Information Science and Technology, Kyushu University}

\affil[2]{Department of Informatics, Kyushu University}

\date{}
\maketitle

\begin{abstract}
We consider labeled directed graphs where each vertex is labeled with a non-empty string. Such labeled graphs are also known as non-linear texts in the literature.
In this paper, we introduce a new problem of comparing two given
labeled graphs, called the SEQ-IC-LCS problem on labeled graphs.
The goal of SEQ-IC-LCS is to compute the the length of the longest common subsequence (LCS) $Z$ of two target labeled graphs $G_1 = (V_1, E_1)$ and $G_2 = (V_2, E_2)$ that includes some string in the constraint labeled graph $G_3 = (V_3, E_3)$ as its subsequence.
Firstly, we consider the case where $G_1$, $G_2$ and $G_3$ are all acyclic,
and present algorithms for computing their SEQ-IC-LCS
in $O(|E_1||E_2||E_3|)$ time and $O(|V_1||V_2||V_3|)$ space.
Secondly, we consider the case where $G_1$ and $G_2$ can be cyclic and $G_3$ is acyclic, and present algorithms for computing their SEQ-IC-LCS 
in $O(|E_1||E_2||E_3| + |V_1||V_2||V_3|\log|\Sigma|)$ time and $O(|V_1||V_2||V_3|)$ space, where $\Sigma$ is the alphabet.

\end{abstract}

\section{Introduction}\label{chap:intro}

We consider \emph{labeled (directed) graphs} where each vertex is labeled with a non-empty string.
Such labeled graphs are also known as \emph{non-linear texts} or \emph{hypertexts} in the literature.
Labeled graphs are a natural generalization of usual (unary-path) strings,
which can also be regarded as a compact representation of a set of strings.
After introduced by the Database community~\cite{hypertext},
labeled graphs were then considered by the string matching community~\cite{Manber,park,amir,navarro,EquiGMT19,EquiMT21,Caceres23}.
Recently, graph representations of large-scale string sets appear in the real-world applications including graph databases~\cite{AnglesABHRV17} and pan-genomics~\cite{bbw089}.
For instance, \emph{elastic degenerate strings}~\cite{GrossiILPPRRVV17,AoyamaNIIBT18,BernardiniPPR20,IliopoulosKP21,BernardiniGPPR22}, which recently gain attention with bioinformatics background, can be regarded as a special case of labeled graphs.
In the best case, a single labeled graph can represent exponentially many strings.
Thus, efficient string algorithms that directly work on labeled graphs without
expansion are of significance both in theory and in practice.

Shimohira et al.~\cite{PSC2011-17} introduced the problem of computing
the \emph{longest common subsequence} (\emph{LCS}) of two given labeled graphs, which, to our knowledge, the first and the only known similarity measure of labeled graphs.
Since we can easily convert any labeled graph with string labels
to an equivalent labeled graph with single character labels (see Figure~\ref{fig:nlts}),
in what follows, we evaluate the size of a labeled graph by the number of vertices and edges in the (converted) graph.
Given two labeled graphs $G_1 = (V_1, E_1)$ and $G_2 = (V_2, E_2)$,
Shimohira et al.~\cite{PSC2011-17} showed how to solve the LCS problem on labeled graphs
in $O(|E_1||E_2|)$ time and $O(|V_1||V_2|)$ space when both $G_1$ and $G_2$ are acyclic,
and in $O(|E_1||E_2| + |V_1||V_2|\log|\Sigma|)$ time and $O(|V_1||V_2|)$ space
when $G_1$ and $G_2$ can be cyclic, where $\Sigma$ is the alphabet.
It is noteworthy that their solution is almost optimal
since the quadratic $O((|A||B|)^{1-\epsilon})$-time conditional lower bound~\cite{AbboudBW15,BringmannK15} with any constant $\epsilon > 0$ for the LCS problem on two strings $A,B$ also applies to the LCS problem on labeled graphs.

The \emph{constrained LCS problems} on strings, which were first proposed by Tsai~\cite{CLCS_Tsai_2003} and then extensively studied in the literature~\cite{CLCS_Tsai_2003,SEQICLCS_2004,ArslanE05,SEQECLCS_Chen_2011,STRICLCS_DEOROWICZ_2012,STRECLCS_Wang_2013,YonemotoNIB23}, use a third input string $P$ which introduces a-priori knowledge of the user to the solution string $Z$ to output. 
The task here is to compute the longest common subsequence $Z$ of two target strings $A$ and $B$ that meets the condition w.r.t. $P$, such that
\begin{description}
  \item[STR-IC-LCS:] $Z$ includes (contains) $P$ as substring;
  \item[STR-EC-LCS:] $Z$ excludes (does not contain) $P$ as substring;
  \item[SEQ-IC-LCS:] $Z$ includes (contains) $P$ as subsequence;
  \item[SEQ-EC-LCS:] $Z$ excludes (does not contain) $P$ as subsequence.
\end{description}
While STR-IC-LCS can be solved in $O(|A||B|)$ time~\cite{STRICLCS_DEOROWICZ_2012},
the state-of-the-art solutions to STR-EC-LCS and SEQ-IC/EC-LCS
run in $O(|A||B||P|)$ time~\cite{SEQICLCS_2004,ArslanE05,SEQECLCS_Chen_2011,STRECLCS_Wang_2013}.

In this paper, we consider the SEQ-IC-LCS problems on labeled graphs,
where the inputs are two target labeled graphs
$G_1 = (V_1, E_1)$ and $G_2 = (V_2, E_2)$, and a constraint text $G_3 = (V_3, E_3)$,
and the output is (the length of) a longest common subsequence of
$G_1$ and $G_2$ such that $Z$ includes as subsequence some string that is represented by $G_3$.
Firstly, we consider the case where $G_1$, $G_2$ and $G_3$ are all acyclic,
and present algorithms for computing their SEQ-IC-LCS 
in $O(|E_1||E_2||E_3|)$ time and $O(|V_1||V_2||V_3|)$ space.
Secondly, we consider the case where $G_1$ and $G_2$ can be cyclic and $G_3$ is acyclic, and present algorithms for computing their SEQ-IC-LCS 
in $O(|E_1||E_2||E_3| + |V_1||V_2||V_3|\log|\Sigma|)$ time and $O(|V_1||V_2||V_3|)$ space, where $\Sigma$ is the alphabet.
The time complexities of our algorithms and related work are summarized in Table~\ref{tab:relatedwork}.
Our algorithms for solving SEQ-IC-LCS on labeled graphs
are based on the solutions to SEQ-IC-LCS of usual strings proposed by
Chin et al.~\cite{SEQICLCS_2004}.
We emphasize that a faster $o(|E_1||E_2||E_3|)$-time solution to the SEQ-IC-LCS problems implies a major improvement over the SEQ-IC-LCS problems for strings whose best known solutions require cubic time.

A related work is 
the \emph{regular language constrained sequence alignment} (\emph{RLCSA})
problem~\cite{ARSLAN} for two input strings $A$ and $B$ in which the constraint is given as an NFA.
It is known that this problem can be solved in $O(|A||B||V|^3/\log|V|)$ time~\cite{RLCSA}, where $|V|$ denotes the number of states in the NFA.

\begin{table}[tbh]
  \caption{Time complexities of algorithms for labeled graph/usual string comparisons, for inputs text-1 $G_1 = (V_1, E_1)$, text-2 $G_2 = (V_2, E_2)$, and text-3 $G_3 = (V_3, E_3)$. Here, a string input of length $n$ is regarded as a unary path graph $G = (V, E)$ with $|E| = n$.}
  \label{tab:relatedwork}
  \centerline{
  \begin{tabular}{|wc{22mm}|wc{10mm}|wc{10mm}|wc{10mm}|l|}\hline
    problem & text-1 & text-2 & text-3 & time complexity\\\hline\hline
    \multirow{3}{*}{LCS}
    & string & string & - & $O(|E_1||E_2|)$ \ \cite{Wagner_1974_LCS}\\ \cline{2-5}
    & DAG & DAG & - & $O(|E_1||E_2|)$ \ \cite{PSC2011-17}\\ \cline{2-5}
    & graph & graph & - & $O(|E_1||E_2|+ |V_1||V_2|\log|\Sigma|)$ \ \cite{PSC2011-17}\\\hline\hline
    \multirow{3}{*}{SEQ-IC-LCS}
    & string & string & string & $O(|E_1||E_2||E_3|)$ \ \cite{SEQICLCS_2004,ArslanE05}\\ \cline{2-5}
    & DAG & DAG & DAG & $O(|E_1||E_2||E_3|)$ \ [this work]\\\cline{2-5}
    & graph & graph & DAG & $O(|E_1||E_2||E_3| + |V_1||V_2||V_3|\log|\Sigma|)$ \ [this work]\\\hline\hline
    SEQ-EC-LCS & string & string & string & $O(|E_1||E_2||E_3|)$ \ \cite{SEQECLCS_Chen_2011}\\ \hline \hline
    STR-IC-LCS & string & string & - & $O(|E_1||E_2|)$ \ \cite{STRICLCS_DEOROWICZ_2012}\\ \hline \hline
    STR-EC-LCS & string & string & - & $O(|E_1||E_2|)$ \ \cite{STRECLCS_Wang_2013}\\ \hline \hline
    \multirow{1}{*}{RLCSA} 
    & string & string & NFA & $O(|E_1||E_2| |V_3|^3/\log |V_3|)$ \ \cite{RLCSA}\\\hline
  \end{tabular}
  }
\end{table}

\section{Preliminaries}

\subsection{Strings and Graphs}

Let $\Sigma$ be an alphabet.
An element of $\Sigma^*$ is called a \emph{string}.
The \emph{length} of a string $w$ is denoted by $|w|$. 
The \emph{empty string}, denoted by $\varepsilon$, 
is a string of length $0$.
Let $\Sigma^+ = \Sigma^\ast \setminus \{\varepsilon\}$.
For a string $w = xyz$ with $x,y,z \in \Sigma^*$,
strings $x$, $y$, and $z$ are called a \emph{prefix}, \emph{substring}, 
and \emph{suffix} of string $w$, respectively.
The $i$th character of a string $w$ is denoted by $w[i]$ for $1\le i\le|w|$, and
the substring of $w$ that begins at position $i$ and ends at position $j$
is denoted by $w[i..j]$ for $1\le i\le j\le |w|$.
For convenience, let $w[i..j]=\varepsilon$ for $i>j$.
A string $u$ is a \emph{subsequence} of another string $w$
if $u = \varepsilon$ or there exists a sequence of integers $i_1,\ldots,i_{|u|}$ 
such that $1\le i_1<\cdots<i_{|u|}\le|w|$ and $u=w[i_1]\cdots w[i_{|u|}]$.

A \emph{directed graph} $G$ is an ordered pair $(V,E)$ of the set $V$ of \emph{vertices}
and the set $E \subseteq V \times V$ of \emph{edges}.
The \emph{in-degree} of a vertex $v$ is denoted by $\ind(v) = |\{u \mid (u, v) \in E \}|$.
A \emph{path} in a (directed) graph $G=(V,E)$ is
a sequence $v_0,\ldots,v_k$ of vertices such that $(v_{i-1},v_i)\in E$ 
for every $i=1,\ldots,k$.
A path $\pi = v_0, \ldots, v_k$ in graph $G$ is said to be \emph{left-maximal}
if its left-end vertex $v_0$ has no in-coming edges,
and $\pi$ is said to be \emph{right-maximal} if its right-end vertex $v_k$ has no out-going edges.
A path $\pi$ is said to be \emph{maximal} if $\pi$ is both left-maximal and right-maximal.
For any vertex $v \in V$, let $\P(v)$ denote the set of all paths 
ending at vertex $v$, and $\LMP(v)$ denote the set of left-maximal paths ending at $v$.
The set of all paths in $G = (V, E)$ is denoted by $\P(G) = \{\P(v) \mid v \in V\}$.
Let $\MP(G)$ denote the set of maximal paths in $G$.


\subsection{Longest Common Subsequence (LCS) of Strings}\label{ssec:Lsequence}
The \emph{longest common subsequence} (LCS) problem for two given strings $A$ and $B$ is to compute (the length of) the longest string $Z$ that is a subsequences of both $A$ and $B$.
It is well-known that LCS can be solved in $O(|A||B|)$ time 
by using the following recurrence~\cite{Wagner_1974_LCS}:
\begin{equation*}
\label{eq:Lsequence}
C_{i,j}= 
\begin{cases}
0                     & \mbox{if $i=0$ or $j=0$}; \\
1+C_{i-1,j-1}            & \mbox{if $i,j>0$ and $x[i]=y[j]$}; \\
\max(C_{i-1,j},C_{i,j-1}) & \mbox{if $i,j>0$ and $x[i]\not=y[j]$},\\
\end{cases}
\end{equation*}
where $C_{i,j}$ is the LCS length of $A[1..i]$ and $B[1..j]$.

\subsection{SEQ-IC-LCS of Strings}\label{ssec:SEQ-I-L}
Let $A$, $B$, and $P$ be strings. A string $Z$ is said to be an \emph{SEQ-IC-LCS} 
of two target strings $A$ and $B$ \emph{including} the pattern $P$ if $Z$ is a longest string 
such that $P$ is a subsequence of $Z$ and that $Z$ is a common subsequence of $A$ and $B$.
Chin et al.~\cite{SEQICLCS_2004} solved this problem in $O(|A||B||P|)$ time by using the following recurrence:
\begin{equation}
   C_{i,j,k}= 
   \begin{cases}
      0                              & \mbox{if $k=0$ and ($i=0$ or $j=0$)}; \\
      -\infty                        & \mbox{if $k \neq 0$ and ($i=0$ or $j=0$)}; \\
      C_{i-1,j-1,k-1}+1              & \mbox{if $i,j,k > 0$ and $A[i] = B[j] = P[k]$}; \\
      C_{i-1,j-1,k}+1                & \mbox{if $i,j > 0$ and $A[i] = B[j] \neq P[k]$}; \\
      \max(C_{i-1,j,k}, C_{i,j-1,k}) & \mbox{if $i,j > 0$ and $A[i] \neq B[j]$}, \\
   \end{cases}
   \label{eq:SEQ-I-L}
\end{equation}
where $C_{i,j,k}$ is the SEQ-IC-LCS length of $A[1..i]$, $B[1..j]$, and $P[1..k]$.


\subsection{Labeled Graphs}
A \emph{labeled graph} is a directed graph with vertices labeled by strings,
namely, it is a directed graph $G = (V, E, \lf)$ where
$V$ is the set of vertices, $E$ is the set of edges,
and $\lf : V \rightarrow \Sigma^+$ is a labeling function
that maps nodes $v \in V$ to non-empty strings $\lf(v)\in\Sigma^+$.
For a path $\pi =v_0,\ldots,v_k \in \P(G)$, 
let $\lf(\pi)$ denote the string spelled out by $w$,
namely $\lf(\pi)=\lf(v_0)\cdots \lf(v_k)$.
The size $|G|$ of a labeled graph $G = (V, E, \lf)$ is $|V|+|E|+\sum_{v \in V}|\lf(v)|$.
Let 
$\Subseq(G) = \{\Subseq(\lf(\pi)) \mid \pi \in \P(G)\}$ denote the set of subsequences of a labeled graph $G = (V, E, \lf)$.
For a set $P \in \P(G)$ of paths in $G$, 
let $\lf(P) = \{\lf(\pi) \mid \pi \in P\}$ denote the set of
string labels for the paths in $P$.

For a labeled graph $G = (V,E, \lf)$,
consider an ``atomic'' labeled graph $G^\prime = (V^\prime, E^\prime, \lf^\prime)$ 
such that $\lf^\prime : V^\prime \rightarrow \Sigma$,
\begin{eqnarray*}
V^\prime & = & \{v_{i,j} \mid \lf^\prime(v_{i,j}) = \lf(v_i)[j], v_i \in V, 1 \leq j \leq |\lf(v_i)|\}, \mbox{ and}\\
E^\prime & = & \{(v_{i,|\lf(v_i)|}, v_{k,1}) \mid (v_i, v_k) \in E\} \cup \{(v_{i,j}, v_{i, j+1}) \mid v_i \in V, 1 \leq j < |\lf(v_i)|\},
\end{eqnarray*}
that is, $G^\prime$ is a labeled graph
with each vertex being labeled by a single character,
which represents the same set of strings as $G$.
An example is shown in Figure~\ref{fig:nlts}.
\begin{figure}[t]
 \centering{
    \includegraphics[scale=0.4]{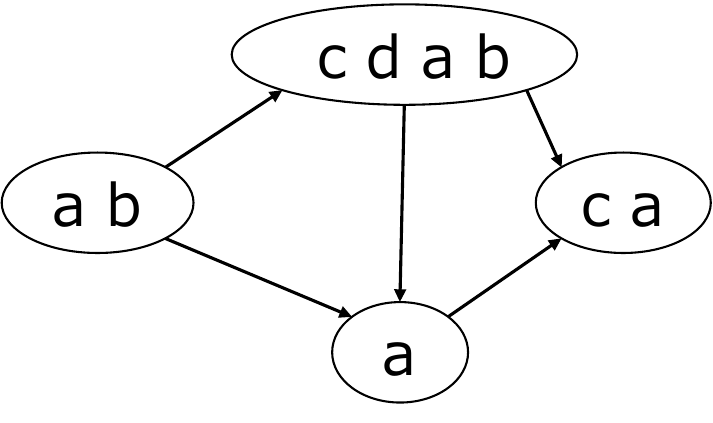}
    \hfil
    \includegraphics[scale=0.4]{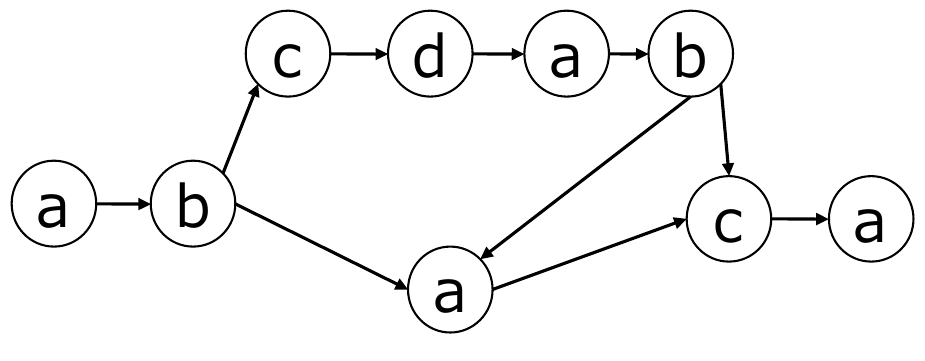}
    \caption{A labeled graph $G = (V, E, \lf)$ with $\lf:V \rightarrow \Sigma^+$ and its corresponding atomic labeled graph $G^\prime = (V^\prime, E^\prime, \lf^\prime)$ with $\lf^\prime: V^\prime \rightarrow \Sigma$.}
    \label{fig:nlts}
 }
\end{figure}
Since $|V^\prime| = \sum_{v \in V}|\lf(v)|$, 
$|E^\prime| = |E| + \sum_{v \in V}(|\lf(v)|-1)$, 
and $\sum_{v^\prime \in V^\prime}|\lf(v^\prime)| = \sum_{v \in V}|\lf(v)|$,
we have $|G^\prime| = O(|G|)$.
We remark that given $G$, we can easily construct $G^\prime$ in $O(|G|)$ time.
Observe that $\Subseq(G) = \Subseq(G^\prime)$ also holds.


In the sequel we only consider atomic labeled graphs 
where each vertex is labeled with a single character.

\subsection{LCS of Acyclic Labeled Graphs}\label{ssc:nl-LCS}
The problem of computing the length of longest common subsequence of two input acyclic labeled graphs is formalized by Shimohira et al.~\cite{PSC2011-17} as follows.

\begin{problem}[Longest common subsequence problem for acyclic labeled graphs]\label{prob:subseq-nl}
\hfill
\begin{description}
\item[Input:] Labeled graphs $G_1=(V_1,E_1,\lf_1)$ and $G_2=(V_2,E_2,\lf_2)$.
\item[Output:] The length of a longest string in $\Subseq(G_1) \cap \Subseq(G_2)$.
\end{description}
\end{problem}

This problem can be solved in $O(|E_1||E_2|)$ time and $O(|V_1||V_2|)$ space by sorting $G_1$ and $G_2$ topologically and using the following recurrence:
\begin{eqnarray} 
   \lefteqn{C^{\prime}_{i,j} =} \nonumber\\ 
      & & \begin{cases}
       1 \! + \! \max(\{ C^{\prime}_{k, \ell} \mid {(v_{1,k},v_{1,i}) \! \in \! E_1, (v_{2,\ell},v_{2,j}) \! \in \! E_2} \} \cup \{0\}) & \mbox{if $\lf_1(v_{1,i}) \! = \! \lf_2(v_{2,j})$};\\
       \max\!\left( 
       \begin{array}{l}
         \{ C^{\prime}_{k,j}  \mid {(v_{1,k},v_{1,i})\! \in \! E_1}\} \cup{} \\
         \{ C^{\prime}_{i,\ell} \mid {(v_{2,\ell},v_{2,j})\! \in \! E_2}\} \cup \{0\}
       \end{array} \right) 
       & \mbox{otherwise},
      \end{cases} \label{eq:subseq-nl}
\end{eqnarray}
where $v_{1,i}$ and $v_{2,j}$ are respectively the $i$th and $j$th vertices of $G_1$ and in $G_2$ in topological order, 
for $1 \leq i \leq |V_1|$ and $1 \leq j \leq |V_2|$, 
and $C^{\prime}_{i,j}$ is the length of a longest string in $\Subseq(\lf_1(\P(v_{1,i}))) \cap \Subseq(\lf_2(\P(v_{2,j})))$. 

\subsection{LCS of Cyclic Labeled Graphs}\label{ssc:nl-LCS-ac}
Here we consider a generalized version of Problem~\ref{prob:subseq-nl} where the input labeled graphs $G_1$ and/or $G_2$ can be cyclic.
In this problem, the output is $\infty$ if there is a string $s \in \Subseq(G_1) \cap \Subseq(G_2)$ such that $|s| = \infty$,
and that is the length of a longest string in $\Subseq(G_1) \cap \Subseq(G_2)$.
Shimohira et al.~\cite{PSC2011-17} proposed an $O(|E_1||E_2| +  |V_1||V_2|\log{|\Sigma|})$ time and $O(|V_1||V_2|)$ space algorithm solving this problem. 
Their algorithm judges whether the output is $\infty$ by using a balanced tree, and computes the length of the solution by using Equation~(\ref{eq:subseq-nl}) and the balanced tree if the output is not $\infty$.

\section{The SEQ-IC-LCS Problem for Labeled Graphs}\label{chap:SEQ-I-L-nl}

In this paper, we tackle the problem of computing the 
SEQ-IC-LCS length of three labeled graphs, which formalized as follows:

\begin{problem}[SEQ-IC-LCS problem for labeled graphs]\label{prob:SEQ-I-L-nl}
  \hfill
  \begin{description}
  \item[Input:] Labeled graphs $G_1=(V_1,E_1,\lf_1)$, $G_2=(V_2,E_2,\lf_2)$, and $G_3=(V_3,E_3,\lf_3)$.
  \item[Output:] The length of a longest string in the set \\ $\{z \mid \exists~q \in \lf_3(\MP(G_3)) \mbox{ such that } q \in \Subseq(z) \mbox{ and } z \in \Subseq(G_1) \cap \Subseq(G_2)\}$.
  \end{description}
\end{problem}

Intuitively, Problem~\ref{prob:SEQ-I-L-nl} asks to compute
a longest string $z$ such that $z$ is a subsequence occurring in both $G_1$ and $G_2$
and that there exists a string $q$ which corresponds to a maximal path of $G_3$
and is a subsequence of $z$.

For a concrete example, see the labeled graphs $G_1$, $G_2$ and $G_3$ of Figure~\ref{fig:SEQ-I-L-ac}.
String $\mathtt{cdba}$ is a common subsequence of $G_1$ and $G_2$
and that contains an element $\mathtt{ba}$ of a maximal path string in $\lf_3(\MP(G_3))$.
Since $\mathtt{cdba}$ is such a longest string, we ouput the SEQ-IC-LCS length $|\mathtt{cdba}| = 4$ as the solution to this instance.

In the sequel,
Section~\ref{sec:SEQ-I-L-ac} presents our solution to the case where the all input labeled graphs are acyclic,
and Section~\ref{sec:SEQ-I-L-c} presents our solutions case where $G_1$ and/or $G_2$ can be cyclic and $G_3$ is acyclic.

\section{Computing SEQ-IC-LCS of Acyclic Labeled Graphs}\label{sec:SEQ-I-L-ac}

In this section, we present our algorithm which solves Problem~\ref{prob:SEQ-I-L-nl} in the case where all of $G_1$, $G_2$ and $G_3$ are acyclic.
The following is our result:

\begin{theorem}\label{theo:SEQ-I-L-nl}
  Problem~\ref{prob:SEQ-I-L-nl} where input labeled graphs $G_1$, $G_2$ and $G_3$ are all acyclic 
  is solvable in $O(|E_1||E_2||E_3|)$ time and $O(|V_1||V_2||V_3|)$ space.
\end{theorem}

\begin{proof}
We perform topological sort to the vertices of $G_1$, $G_2$, and $G_3$ in
$O(|E_1|+|E_2|+|E_3|)$ time and $O(|V_1|+|V_2|+|V_3|)$ space.
For $1 \leq i \leq |V_1|$, $1 \leq j \leq |V_2|$, and $1 \leq k \leq |V_3|$,
let $v_{1,i}$, $v_{2,j}$, $v_{3,k}$ denote the $i$th, $j$th, and $k$th
vertices in $G_1$, $G_2$, and $G_3$ in topological order,
respectively.
Let
\begin{equation*}
\mathsf{S}_{\mathrm{IC}}(v_{1,i}, v_{2,j}, v_{3,k}) = 
\left\{
z ~\middle|~
\begin{array}{l}
  \exists q \in \lf_3(\LMP(v_{3,k})) \mbox{ such that } q \in \Subseq(z) \\
  \mbox{and } z \in \Subseq(\lf_1(\P(v_{1,i}))) \cap \Subseq(\lf_2(\P(v_{2,j})))
\end{array}
\right\}
\end{equation*}
be the set of candidates of SEQ-IC-LCS strings for the maximal induced graphs
of $G_1$, $G_2$, and $G_3$ whose sinks are $v_{1,i}$, $v_{2,j}$, and $v_{3,k}$,
respectively.
Let $D_{i,j,k}$ denote the length of a longest string in $\mathsf{S}_{\mathrm{IC}}(v_{1,i}, v_{2,j}, v_{3,k})$.
The solution to Problem~\ref{prob:SEQ-I-L-nl} (the SEQ-IC-LCS length) is the maximum value of $D_{i,j,k}$ for which $v_{3,k}$ does not have out-going edges (i.e. $v_{3,k}$ is the end of a maximal path in $G_3$).

When $k=0$, then the problem is equivalent to Problem~\ref{prob:subseq-nl} of computing SEQ-IC-LCS of strings.
In that follows, we show how to compute $D_{i,j,k}$ for $k > 0$:
\begin{enumerate}
  \item If $\lf_1(v_{1, i}) = \lf_2(v_{2, j}) = \lf_3(v_{3, k})$, there are three cases to consider:
  \begin{enumerate}
    \item If \sinote{modified}{$v_{1, i}$ does not have in-coming edges or $v_{2, j}$ does not have in-coming edges,
          and if $v_{3, k}$ does not have in-coming edges}
          (i.e., $\ind(v_{1, i}) =  \ind(v_{3,k}) = 0$, or $\ind(v_{2, j}) = \ind(v_{3, k}) = 0$),
          then clearly $D_{i, j, k} = 1$.
    \item If \sinote*{modified}{$v_{1, i}$ does not have in-coming edges or $v_{2, j}$ does not have in-coming edges,
          and if $v_{3, k}$ has some in-coming edge(s)}
          (i.e., $\ind(v_{1, i}) = 0$ and $\ind(v_{3, k}) \geq 1$, or $\ind(v_{2, j}) = 0$ and $\ind(v_{3, k}) \geq 1$),
          then clearly $D_{i, j, k} = -\infty$.
    \item If both $v_{1, i}$ and $v_{2, j}$ have some in-coming edge(s)
          and $v_{3, k}$ does not have in-coming edges
          (i.e., $\ind(v_{1,i}) \geq 1$, $\ind(v_{2,j}) \geq 1$,
          and $\ind(v_{3, k}) = 0$),
          then 
          let $v_{1,x}$ and $v_{2,y}$ be any nodes s.t.
          $(v_{1,x}, v_{1, i}) \in E_1$, and $(v_{2, y}, v_{2, j}) \in E_2$, 
          respectively.
          Let $s$ be a longest string in 
          $\Subseq(\lf_1(\P(v_{1,i}))) \cap \Subseq(\lf_2(\P(v_{2,j})))$.
          Assume on the contrary that 
          there exists a string $t \in \Subseq(\lf_1(\P(v_{1,x}))) \cap 
          \Subseq(\lf_2(\P(v_{2,y})))$ such that $|t| > |s| - 1$.
          This contradicts that $s$ is a longest common subsequence of 
          $\lf_1(\P(v_{1,i}))$ and $\lf_2(\P(v_{2,j}))$,
          since $\lf_1(v_{1, i}) = \lf_2(v_{2,j})$.
          Hence $|t| \leq |s| - 1$. 
          If $v_{1, x}$ and $v_{2, y}$ are vertices
          satisfying $C^\prime_{x, y, 0} = |s| - 1$, then $C^\prime_{i, j, k} = C^\prime_{x, y, 0} + 1$.
          Note that such nodes
          $v_{1, x}$ and $v_{2, y}$
          always exist.
    \item Otherwise (all $v_{1,i}$, $v_{2,j}$, and $v_{3,k}$ have some in-coming edge(s)), 
          let $v_{1,x}$, $v_{2,y}$ and $v_{3,z}$ be any nodes s.t.
          $(v_{1,x}, v_{1, i}) \in E_1$, $(v_{2, y}, v_{2, j}) \in E_2$ and $(v_{3, z}, v_{3, k}) \in E_3$, 
          respectively.
          Let $s$ be a longest string in $\mathsf{S}_{\mathrm{IC}}(v_{1,i}, v_{2,j}, v_{3,k})$.
          Assume on the contrary that 
          there exists a string $t \in \mathsf{S}_{\mathrm{IC}}(v_{1,x}, v_{2,y}, v_{3,z})$ such that $|t| > |s|-1$.
          This contradicts that $s$ is a SEQ-IC-LCS of 
          $\lf_1(\P(v_{1,i}))$, $\lf_2(\P(v_{2,j}))$ and $\lf_3(\LMP(v_{3,k}))$,
          since $\lf_1(v_{1, i}) = \lf_2(v_{2,j}) = \lf_3(v_{3,k})$.
          Hence $|t| \leq |s| - 1$. 
          If $v_{1, x}$, $v_{2, y}$ and $v_{3, z}$ are vertices
          satisfying $D_{x, y, z} = |s| - 1$, then
          $D_{i, j, k} = D_{x, y, z} + 1$.
          Note that such nodes $v_{1, x}$, $v_{2, y}$ and $v_{3, z}$ always exist.
  \end{enumerate}

\item If $\lf_1(v_{1, i}) = \lf_2(v_{2, j}) \neq \lf_3(v_{3, k})$, there are two cases to consider:
   \begin{enumerate}
    \item If \sinote*{modified}{$v_{1, i}$ does not have in-coming edges or $v_{2, j}$ does not have-incoming edges}
          (i.e., $\ind(v_{1, i}) = 0$ or $\ind(v_{2, j}) = 0$),
           then clearly $D_{i, j, k}$ does not exist and let $D_{i, j, k} = -\infty$. 
    \item Otherwise (both $v_{1, i}$ and $v_{2, j}$ have in-coming edge(s)),
           let $v_{1,x}$ and $v_{2,y}$ be any nodes s.t.
           $(v_{1,x}, v_{1, i}) \in E_1$ and $(v_{2, y}, v_{2, j}) \in E_2$, 
           respectively.
           Let $s$ be a longest string in $\mathsf{S}_{\mathrm{IC}}(v_{1,i}, v_{2,j}, v_{3,k})$.
           Assume on the contrary that 
           there exists a string $t \in \mathsf{S}_{\mathrm{IC}}(v_{1,x}, v_{2,y}, v_{3,k})$ such that $|t| > |s|-1$.
           This contradicts that $s$ is a SEQ-IC-LCS of 
           $\lf_1(\P(v_{1,i}))$, $\lf_2(\P(v_{2,j}))$ and $\lf_3(\LMP(v_{3,k}))$,
           since $\lf_1(v_{1, i}) = \lf_2(v_{2,j})$.
           Hence $|t| \leq |s| - 1$. 
           If $v_{1, x}$, $v_{2, y}$ and $v_{3, k}$ are vertices
           satisfying $D_{x, y, k} = |s| - 1$, then $D_{i, j, k} = D_{x, y, k} + 1$.
           Note that such nodes $v_{1, x}$, $v_{2, y}$ and $v_{3, k}$ always exist.
    \end{enumerate}
           
\item If $\lf_1(v_{1, i}) \neq \lf_2(v_{2, j})$, there are two cases to consider:
    \begin{enumerate}
      \item If \sinote*{modified}{$v_{1, i}$ does not have in-coming edges and $v_{2, j}$ does not have in-coming edges}
           (i.e., $\ind(v_{1, i}) = \ind(v_{2, j}) = 0$),
           then clearly $D_{i, j, k}$ does not exist and let $D_{i, j, k} = -\infty$.
      \item Otherwise (\sinote*{modified}{$v_{1, i}$ has some in-coming edge(s) or $v_{2, j}$ has some in-coming edge(s)}),
            let $v_{1,x}$ and $v_{2, y}$ be any nodes such that
            $(v_{1,x}, v_{1, i}) \in E_1$ and $(v_{2, y}, v_{2, j}) \in E_2$, 
            respectively.
            Let $s$ be a $\mathsf{S}_{\mathrm{IC}}(v_{1,i}, v_{2,j}, v_{3,k})$.
            Assume on the contrary that there exists a string 
            $t \in \mathsf{S}_{\mathrm{IC}}(v_{1,i}, v_{2,j}, v_{3,k})$
            such that $|t| > |s|$.
            This contradicts that $s$ is a SEQ-IC-LCS of 
            $\lf_1(\P(v_{1,i}))$, $\lf_2(\P(v_{2,j}))$ and $\lf_3(\LMP(v_{3,k}))$, 
            since  $\mathsf{S}_{\mathrm{IC}}(v_{1,x}, v_{2,y}, v_{3,k}) \subseteq  \mathsf{S}_{\mathrm{IC}}(v_{1,i}, v_{2,j}, v_{3,k})$.
            Hence $|t| \leq |s|$.
            If $v_{1, x}$ is a vertex satisfying $D_{x, j, k} = |z|$, then $D_{i,j,k} = D_{x,j,k}$.
            Similarly, if $v_{2,y}$ is a vertex satisfying $D_{i,y,k} = |s|$,
            then $D_{i,j,k} = D_{i,y,k} $.
            Note that such node $v_{1,x}$ or $v_{2,y}$ always exists.
   \end{enumerate}
\end{enumerate}

Consequently we obtain the following recurrence:
\begin{eqnarray} 
\lefteqn{D_{i,j,k} =} \nonumber\\ 
 & & \begin{cases}
  \mbox{Recurrence in Equation}~(\ref{eq:subseq-nl}) & \mbox{if $k=0$};\\ 
  1 + \max \left(
  \left \{D_{x,y,z} ~\middle|~
  \begin{array}{l}
  (v_{1,x},v_{1,i}) \in E_1, \\
  (v_{2,y},v_{2,j}) \in E_2, \\
  (v_{3,z},v_{3,k}) \in E_3, \\
  \mbox{ or } z = 0 
  \end{array} \right\}
  \cup \{\gamma \} \right) 
  & 
  \begin{array}{l}
    \mbox{if } k > 0 \mbox{ and} \\
    \lf_1(v_{1,i}) = \lf_2(v_{2,j}) \\
    = \lf_3(v_{3,k});
  \end{array} \\
  \max\left( 
  \left\{1+D_{x,y,k} ~\middle|~
  \begin{array}{l}
    (v_{1,x},v_{1,i}) \! \in \! E_1, \\
    (v_{2,y},v_{2,j}) \! \in \! E_2
  \end{array}
  \right\} 
  \cup \{-\infty\} \right) 
  &
  \begin{array}{l}
    \mbox{if } k > 0 \mbox{ and} \\
    \lf_1(v_{1,i}) = \lf_2(v_{2,j}) \\
    \neq \lf_3(v_{3,k});
  \end{array} \\
  \max\left( 
  \begin{array}{l}
    \{D_{x,j,k} \mid (v_{1,x},v_{1,i}) \in E_1\} \cup{} \\
    \{D_{i,y,k} \mid {(v_{2,y},v_{2,j}) \in E_2}\} 
    \cup \{-\infty\}
  \end{array}
  \right) 
  & \mbox{otherwise}.
     \end{cases}
\label{eqn:SEQ-I-L-nl}
\end{eqnarray}
where
\begin{eqnarray*}
\gamma & = &
\begin{cases}
  0 &
  \begin{array}{l}
    \mbox{if $v_{1, i}$ does not have in-coming edges at all or $v_{2, j}$ does not have} \\
    \mbox{in-coming edges at all, and $v_{3, k}$ does not have in-coming edges};
  \end{array}  \\
  -\infty & \mbox{otherwise}. 
\end{cases}
\end{eqnarray*}

We compute $D_{i,j,k}$ for all $1 \leq i \leq |V_1|$, $1 \leq j \leq |V_2|$ and $0 \leq k \leq |V_3|$,
using a dynamic programming table of size $O(|V_1||V_2||V_3|)$.

Below we analyze the time complexity for computing $D_{i,j,k}$ with the recurrence:
\begin{itemize}
\item The first case with Equation~(\ref{eq:subseq-nl}) takes $O(|E_1||E_2|)$ time (Section~\ref{ssc:nl-LCS}).

\item Second, let us analyze the time cost for computing 
\[
M_{i,j,k} = \max \{ D_{x,y,z} \mid (v_{1,x},v_{1,i}) \in E_1 , (v_{2,y},v_{2,j}) \in E_2, (v_{3,z},v_{3,k}) \in E_3, \mbox{ or } z = 0\}
\]
in the second case of the recurrence for all $i,j,k$.
For each fixed pair of $(v_{1, x}, v_{1, i}) \in E_1$ and $(v_{2, y}, v_{2, j}) \in E_2$,
we refer the value of $D_{x, y, z}$ 
for all $1 \leq z < k$ such that $(v_{3, z}, v_{3, k}) \in E_3$,
in $O(|E_3|)$ time.
For each fixed $(v_{1, x}, v_{1, i}) \in E_1$,
we refer the value of $D_{x, y, z}$
for all $1 \leq y < j$ such that $(v_{2, y}, v_{2, j}) \in E_2$ and all $1 \leq z < k$ such that $(v_{3, z}, v_{3, k}) \in E_3$,
in $O(|E_2||E_3|)$ time.
Therefore, the total time complexity for computing all $M_{i,j,k}$ for all $i,j,k$ is $O(|E_1| |E_2| |E_3|)$.

\item Third, let us analyze the time cost for computing 
\[
M'_{i,j,k} = \max \{ D_{x,y,k} \mid (v_{1,x},v_{1,i}) \in E_1 , (v_{2,y},v_{2,j}) \in E_2\}
\]
in the third case of the recurrence for all $i,j,k$.
For each fixed pair of $(v_{1, x}, v_{1, i}) \in E_1$ and $(v_{2, y}, v_{2, j}) \in E_2$,
we refer the value of $D_{x, y, k}$ 
for all $1 \leq k \leq |V_3|$,
in $O(|V_3|)$ time.
For each fixed $(v_{1, x}, v_{1, i}) \in E_1$,
we refer the value of $D_{x, y, k}$ 
for all $1 \leq y < j$ such that $(v_{2, y}, v_{2, j}) \in E_2$ and all $1 \leq k \leq |V_3|$,
in $O(|E_2||V_3|)$ time.
Therefore, the total time complexity for computing $M'_{i,j,k}$ for all $i,j,k$ 
is $O(|E_1|  |E_2| |V_3|) \subseteq O(|E_1||E_2||E_3|)$.

\item Fourth, let us analyze the time cost for computing 
\[
M''_{i,j,k} = \max \{ D_{x,j,k} ,D_{i,y,k} \mid (v_{1,x},v_{1,i}) \in E_1 , (v_{2,y},v_{2,j}) \in E_2 \}
\]
in the fourth case of the recurrence for all $i,j,k$.
For each fixed $(v_{1,x},v_{1,i})\in E_1$, 
we refer the value of $D_{x,j,k}$ for all $1 \leq j \leq |V_2|$ and all $1 \leq k \leq |V_3|$ in $O(|V_2||V_3|)$ time.
Similarly, for each fixed $(v_{2,y},v_{2,j}) \in E_2$, 
we refer the value of $D_{i,y,k}$ for all $1 \leq i \leq |V_1|$ and all $1 \leq k \leq |V_3|$ in $O(|V_1||V_3|)$ time.
Therefore, the total time cost for computing $M''_{i,j,k}$ for all $i,j,k$ is $O(|V_3|(|V_2|  |E_1| + |V_1|  |E_2|)) \subseteq O(|E_1||E_2||E_3|)$.

\end{itemize}
Thus the total time complexity is $O(|E_1|  |E_2|  |E_3|)$.
\qed
\end{proof}

An example of computing $D_{i,j,k}$ using dynamic programming is show in Figure~\ref{fig:SEQ-I-L-ac}.
We remark that the recurrence in Equation~(\ref{eqn:SEQ-I-L-nl}) is a natural generalization of the recurrence in Equation~(\ref{eq:SEQ-I-L}) for computing the SEQ-IC-LCS length of given two strings.

Algorithm~\ref{algo:SEQ-I-L-nl} in Appendix~\ref{sec:SEQ-IC-LCS_acyclic_pseudocode} shows 
a pseudo-code of our algorithm which solves Problem~\ref{prob:SEQ-I-L-nl} 
in the case where all $G_1$, $G_2$ and $G_3$ are acyclic.

\begin{figure}[t]
   \centerline{
   \begin{tabular}{ccc}
   \begin{minipage}{0.33\hsize}
    \begin{center}
     \includegraphics[scale=0.4]{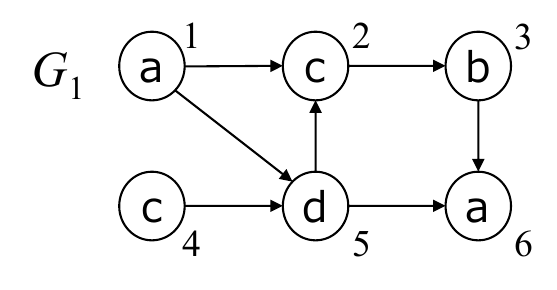}
    \end{center}
    \vspace{-8mm}
    \begin{center}
     \includegraphics[scale=0.4]{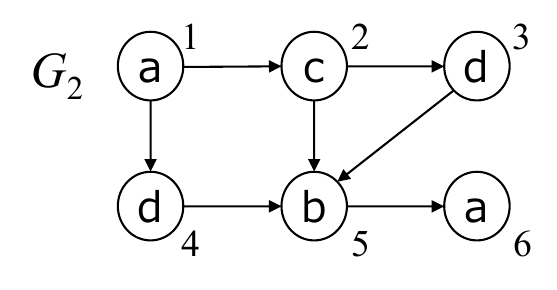}
    \end{center}
    \vspace{-7mm}
    \begin{center}
     \includegraphics[scale=0.4]{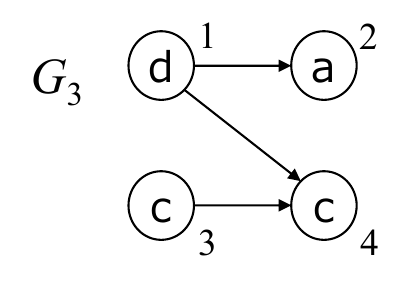}
    \end{center}
    \vspace{-8mm}
    \begin{center}
      \includegraphics[scale=0.4]{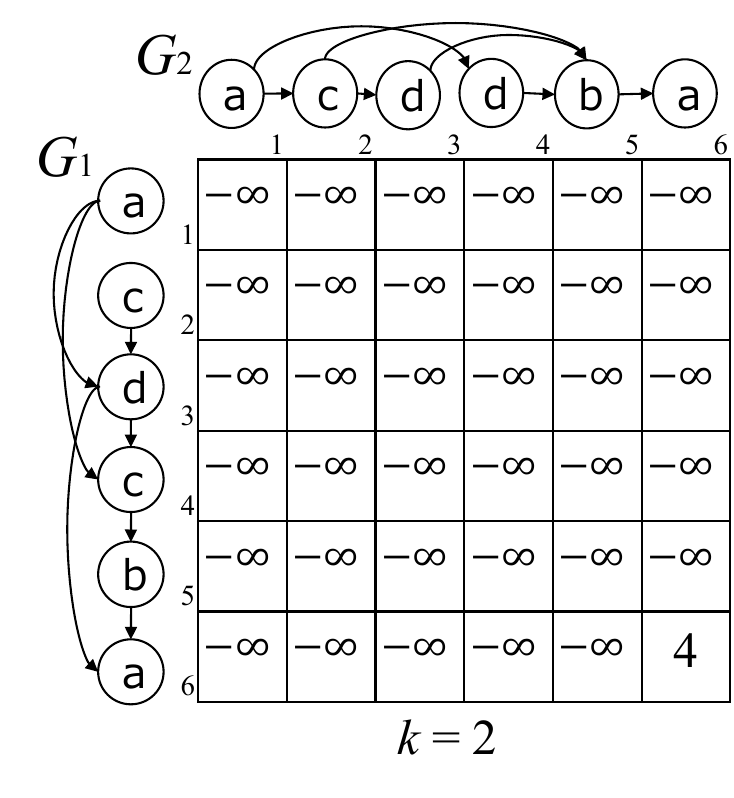}
     \end{center}
   \end{minipage} 
   \begin{minipage}{0.33\hsize}
    \begin{center}
      \includegraphics[scale=0.4]{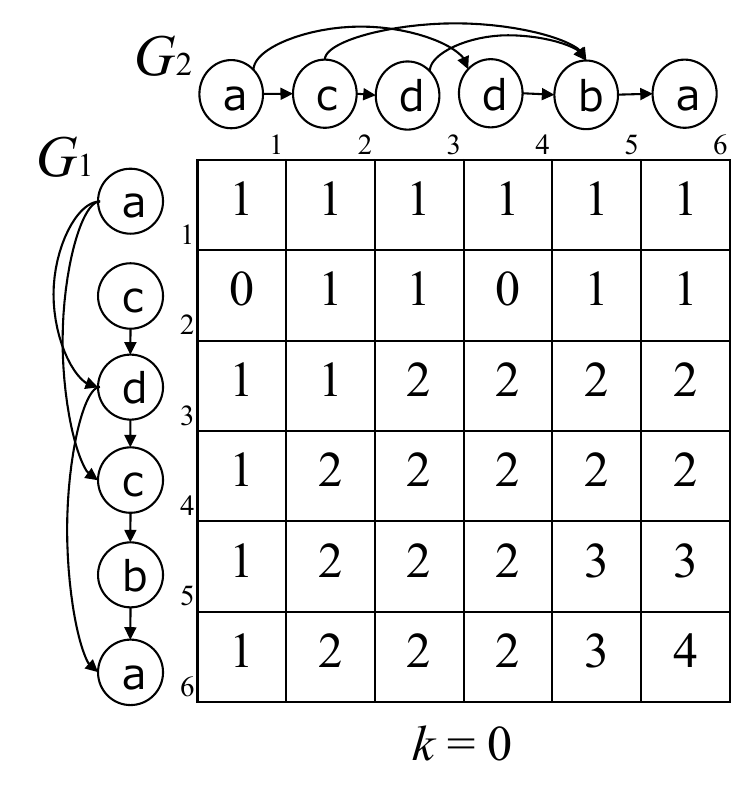}
     \end{center}
     \begin{center}
      \includegraphics[scale=0.4]{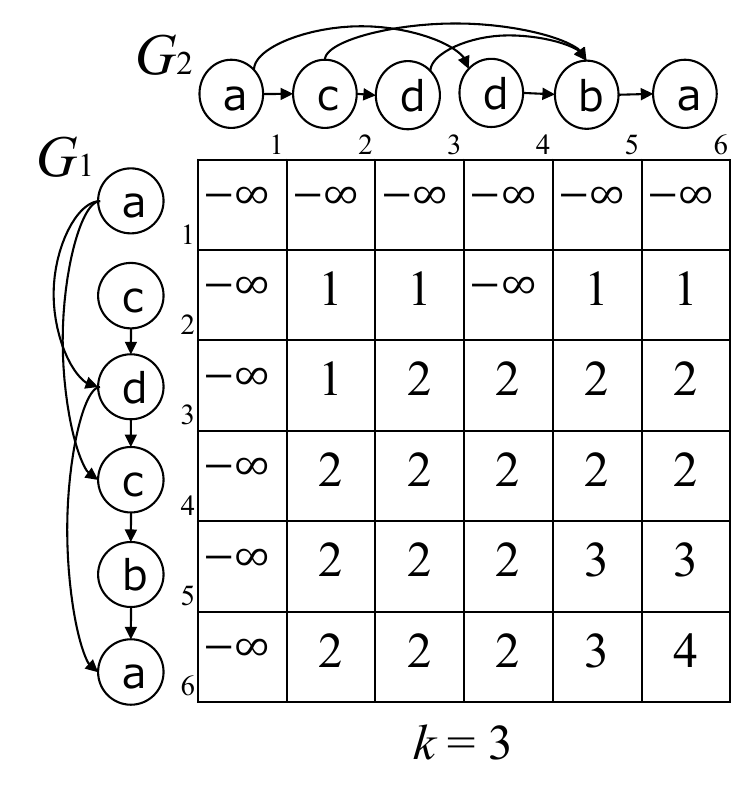}
     \end{center}
    \end{minipage} 
    \begin{minipage}{0.33\hsize}
     \begin{center}
      \includegraphics[scale=0.4]{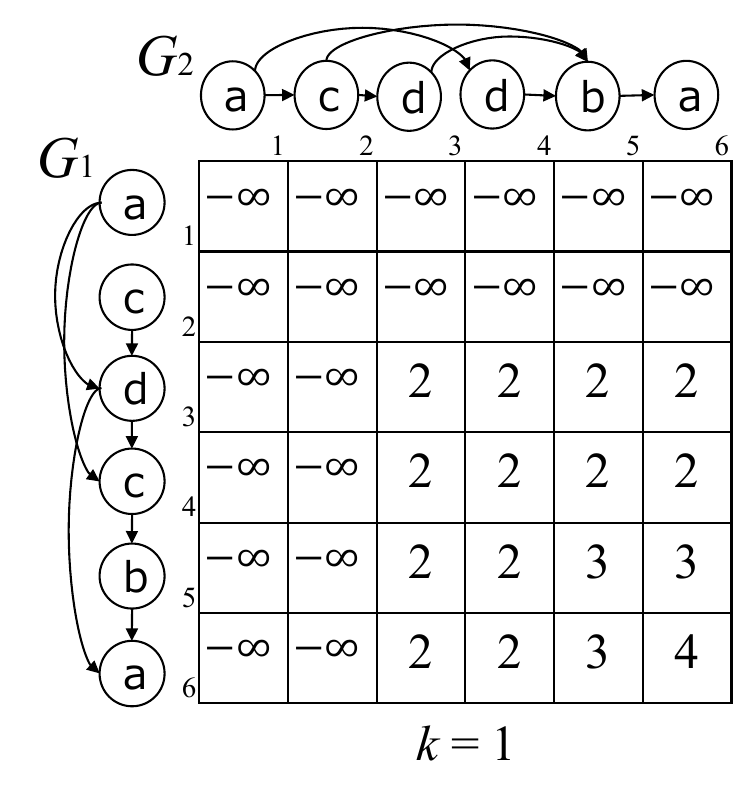}
     \end{center}
     \begin{center}
      \includegraphics[scale=0.4]{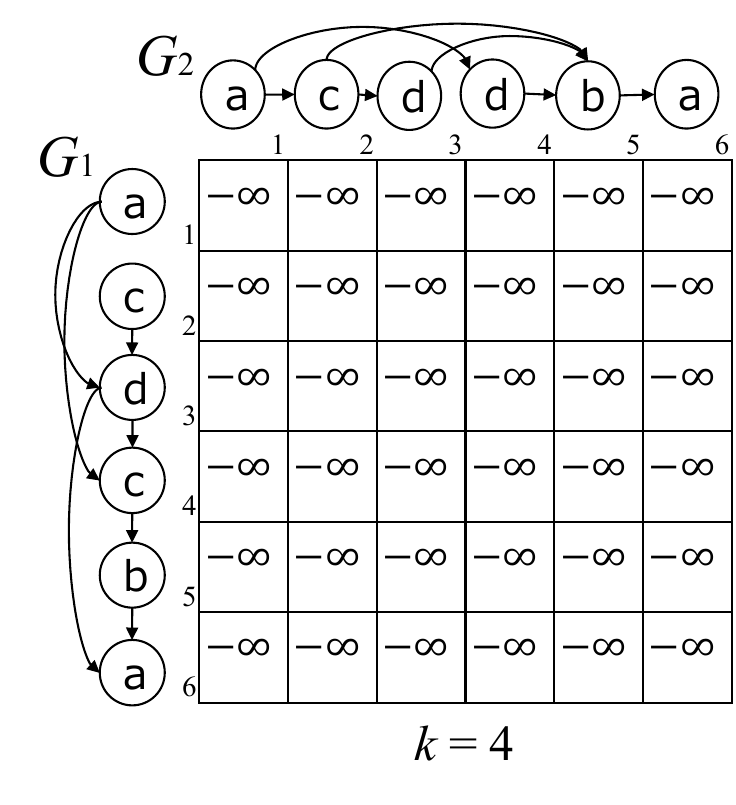}
     \end{center}
   \end{minipage}
   \end{tabular}
   }
    \caption{Example of dynamic programming table $D$ for computing the SEQ-IC-LCS length of
             acyclic labeled graphs $G_1$, $G_2$ and $G_3$. 
             Each vertex is annotated with its topological order.
             In this example, $v_{3,2}$ and $v_{3,4}$ with $k \in \{2,4\}$
             in $G_3$ are vertices with no out-going edges.
             The maximum value of $D_{i,j,k}$ with $k \in \{2,4\}$ is $D_{6,6,2} = 4$,
             and the corresponding SEQ-IC-LCS is $\mathtt{cdba}$ of length $4$.}
    \label{fig:SEQ-I-L-ac}
\end{figure}

\section{Computing SEQ-IC-LCS of Cyclic Labeled Graphs}
\label{sec:SEQ-I-L-c}

In this section, we present an algorithm to solve Problem~\ref{prob:SEQ-I-L-nl} in case where $G_1$ and/or $G_2$ can be cyclic and $G_3$ is acyclic.
We output $\infty$ if the set of output candidates in Problem~\ref{prob:SEQ-I-L-nl} contains a string of infinite length,
and outputs the (finite) SEQ-IC-LCS length otherwise.

\begin{figure}[b!]
\centerline{
\begin{tabular}{ccc}
  \begin{minipage}{0.3\hsize}
   \begin{center}
    \includegraphics[scale=0.35]{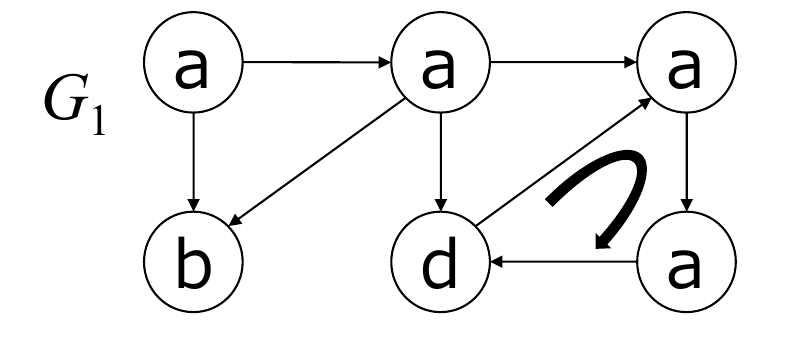}
   \end{center}
   \begin{center}
    \includegraphics[scale=0.35]{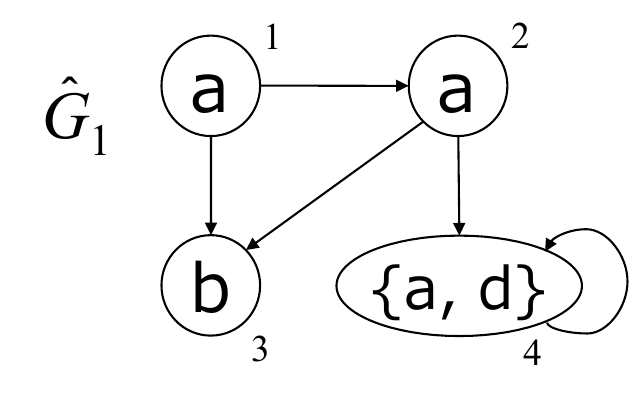}
   \end{center}
   \begin{center}
    \includegraphics[scale=0.35]{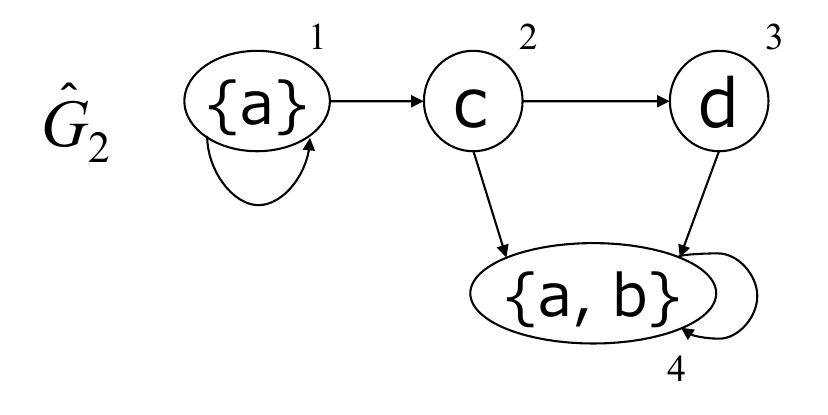}
   \end{center}
   \vspace{-11mm}
   \begin{center}
     \includegraphics[scale=0.35]{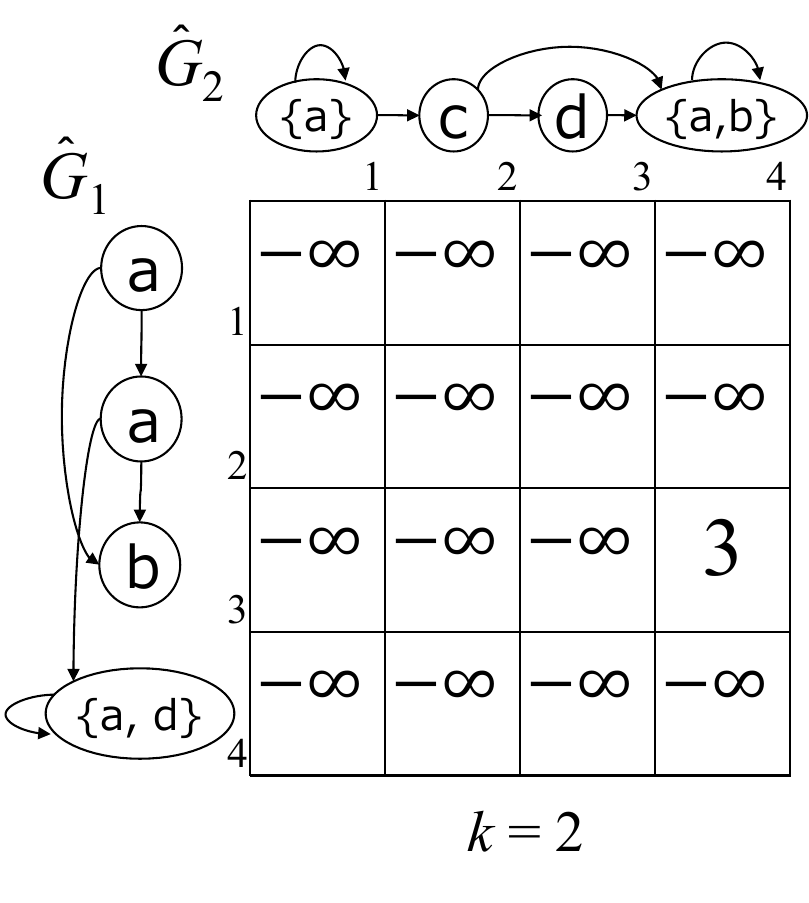}
    \end{center}
  \end{minipage} 
  \begin{minipage}{0.35\hsize}
    \begin{center}
      \includegraphics[scale=0.35]{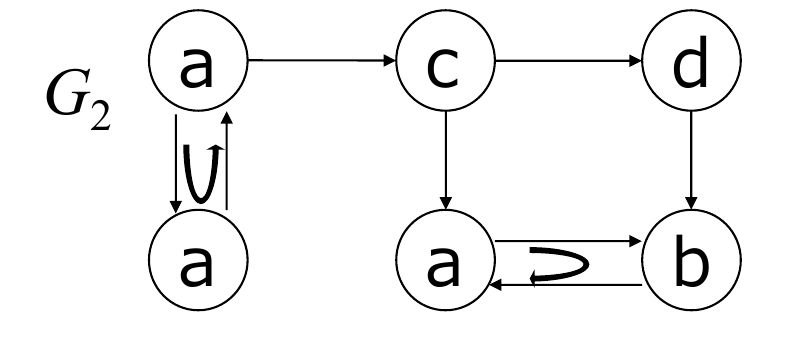}
     \end{center}
   \begin{center}
     \includegraphics[scale=0.35]{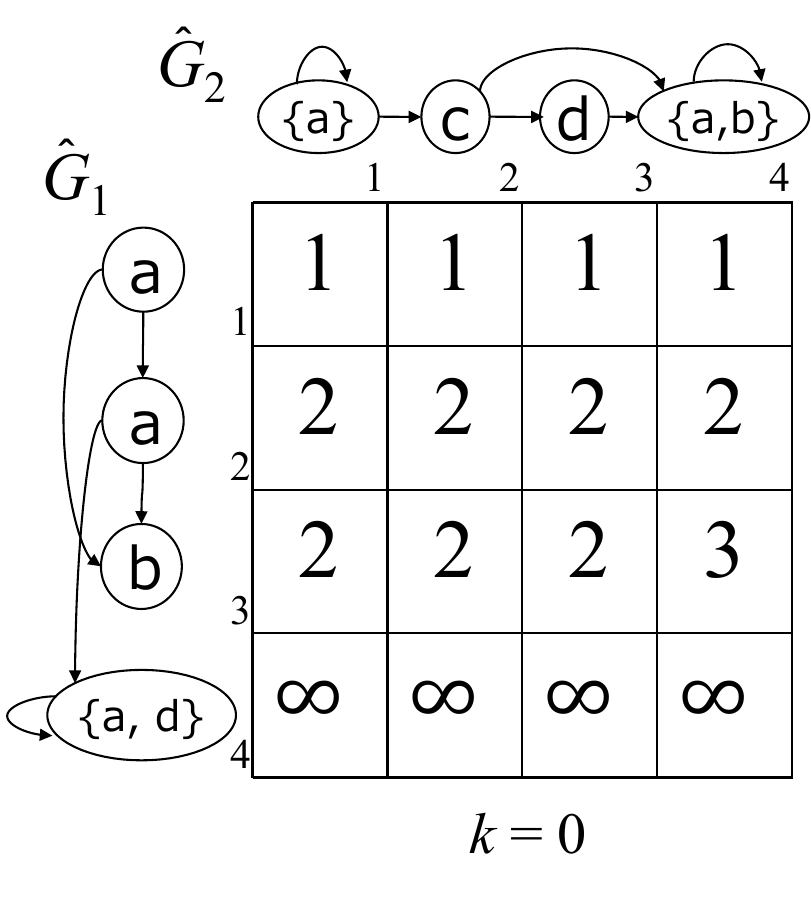}
    \end{center}
    \begin{center}
     \includegraphics[scale=0.35]{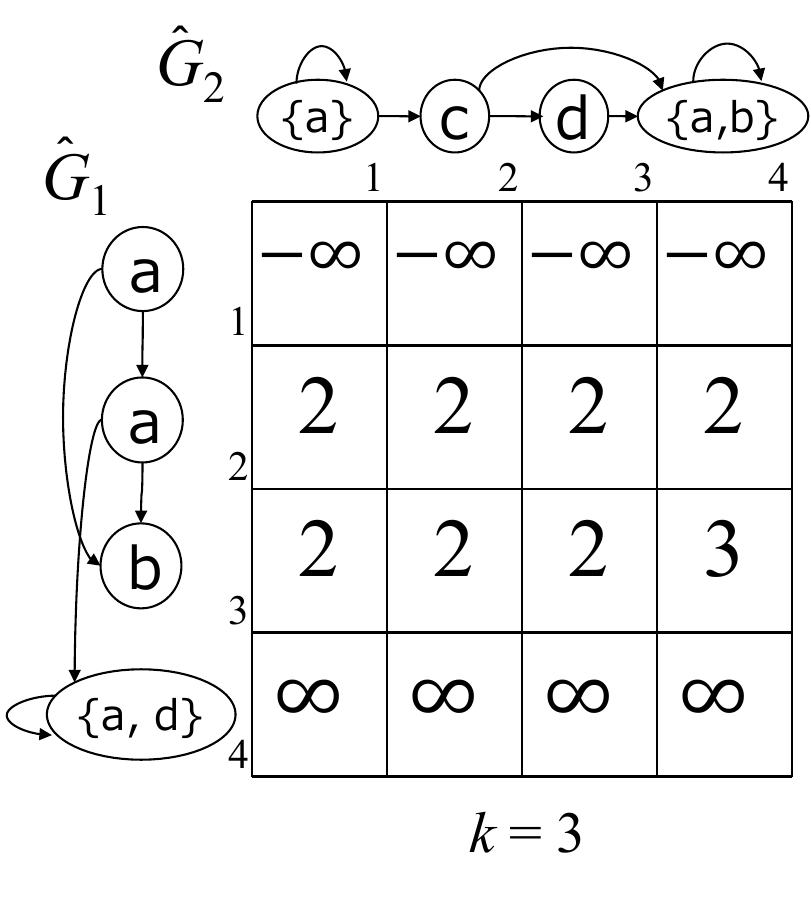}
    \end{center}
   \end{minipage} 
   \begin{minipage}{0.35\hsize}
    \vspace{-1mm}
    \begin{center}
      \includegraphics[scale=0.50]{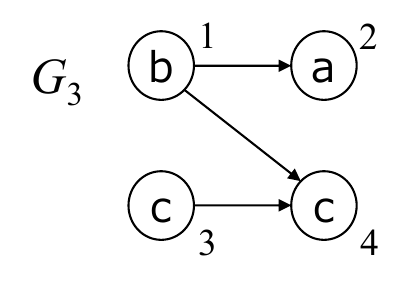}
     \end{center}
     \vspace{-6mm}
    \begin{center}
     \includegraphics[scale=0.35]{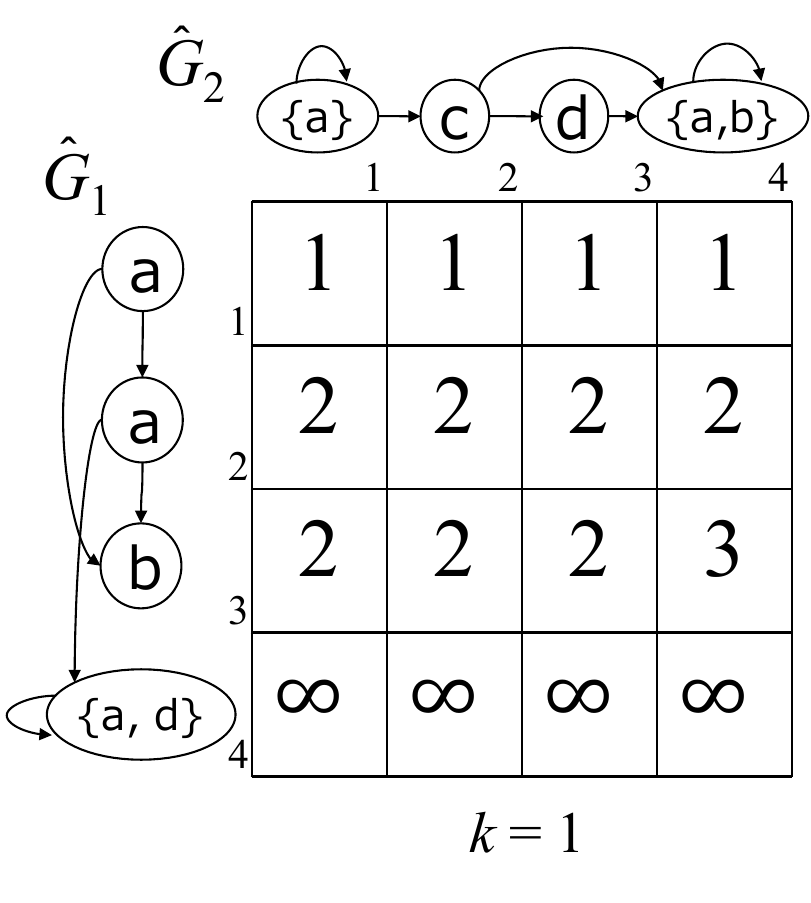}
    \end{center}
    \begin{center}
     \includegraphics[scale=0.35]{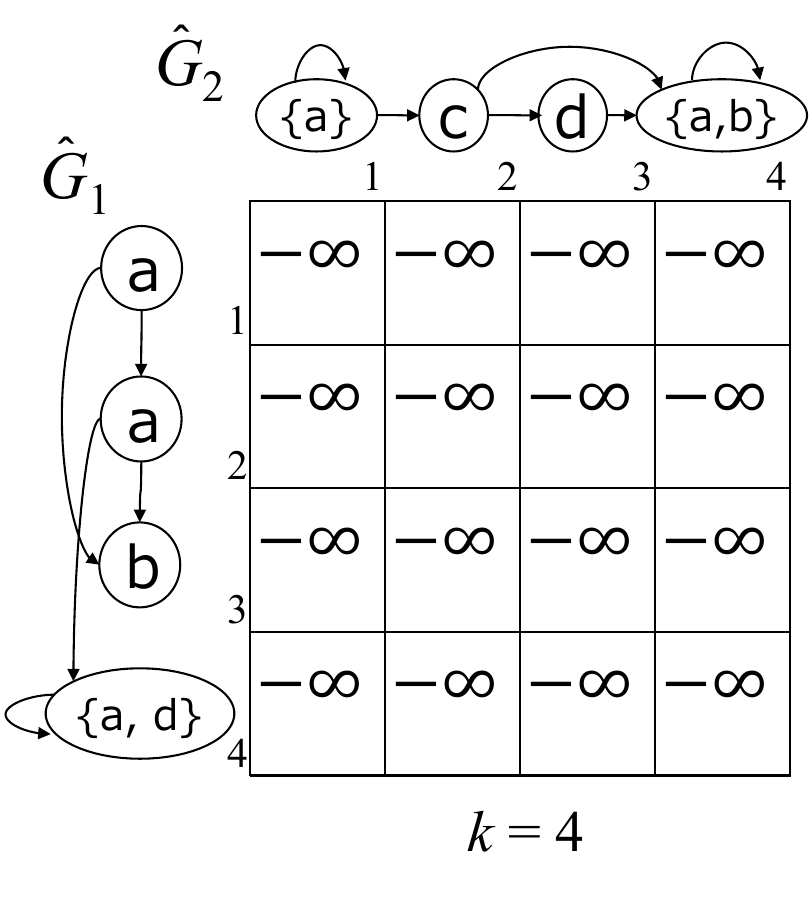}
    \end{center}
  \end{minipage}
\end{tabular}
}
    \caption{Example of dynamic programming table $\hat{D}$ for computing the 
             SEQ-IC-LCS length of cyclic labeled graphs $G_1$ and $G_2$,
             and acyclic labeled graph $G_3$. 
             $\hat{G}_1$ and $\hat{G}_2$ are the labeled graphs which are transformed 
             from $G_1$ and $G_2$ by grouping vertices into strongly connected components.
             Each vertex is annotated with its topological order.
             In this example, $v_{3,2}$ and $v_{3,4}$ with $k \in \{2,4\}$
             in $G_3$ are vertices with no out-going edges.
             The maximum value of $\hat{D}_{i,j,k}$ with $k \in \{2,4\}$ is $\hat{D}_{4,3,2} = 3$,
             and the corresponding SEQ-IC-LCS is $\mathtt{aab}$ of length $3$.             
            }
    \label{fig:SEQ-I-L-c}
\end{figure}

To deal with cyclic graphs, we follow the approach by Shimohira et al.~\cite{PSC2011-17} which transforms a cyclic labeled graph $G = (V, E, \lf)$ 
into an acyclic labeled graph $\hat{G} = (\hat{V}, \hat{E}, \hat{\lf})$
based on the strongly connected components.

For each vertex $v \in V$, let $[v]$ denote the set of vertices that belong to the 
same strongly connected component.
Formally, $\hat{G} = (\hat{V}, \hat{E}, \hat{\lf})$ is defined by
\begin{eqnarray*}
\hat{V} & = & \{[v] \mid v \in V\}, \\
\hat{E} & = & \{([v], [u]) \mid [v] \neq [u], \mbox{$(\hat{v}, \hat{u}) \in E$ for some $\hat{v} \in [v]$, $\hat{u} \in [u]$}\} \cup \{(v, v) \mid |[v]| \geq 2\}, 
\end{eqnarray*}
and $\hat{\lf}([v]) = \{\lf(v) \mid v \in [v]\} \subseteq \Sigma$.
We regard each $[v]$ as a single vertex that is contracted from vertices in $[v]$.
Observe that $\Subseq(\hat{G}) = \Subseq(G)$.
An example of transformed acyclic labeled graphs is shown in Figure~\ref{fig:SEQ-I-L-c}.

\sinote*{added}{%
  It is possible that a vertex $\hat{v} \in \hat{V}$ in
  the transformed graph $\hat{G}$ has a self-loop.
  We regard that a self-loop $(\hat{v}, \hat{v})$ is also an in-coming edge
  of vertex $\hat{v}$.
  We say that vertex $\hat{v}$ does not have in-coming edges \emph{at all},
  if $\hat{v}$ does not have in-coming edges from \emph{any} vertex in $\hat{V}$ (including $\hat{v}$).
}%

Our main result of this section follows:
\begin{theorem}\label{theo:SEQ-I-L-c}
  Problem~\ref{prob:SEQ-I-L-nl}, where input labeled graphs
  $G_1$ and $G_2$ can be cyclic and $G_3$ is acyclic,
  is solvable in $O(|E_1||E_2||E_3| + |V_1||V_2||V_3|\log|\Sigma|)$ 
  time and $O(|V_1||V_2||V_3|)$ space.
\end{theorem}

\begin{proof}
We first transform cyclic labeled graphs $G_1$ and $G_2$ into 
corresponding acyclic labeled graphs $\hat{G}_1$ and $\hat{G}_2$,
as described previously.
For $1 \leq i \leq |\hat{V}_1|$ and $1 \leq j \leq |\hat{V}_2|$,
let $\hat{v}_{1,i}$ and $\hat{v}_{2,j}$ respectively denote the $i$th and $j$th vertices in $\hat{G}_1$ and $\hat{G}_2$ in topological order.
Let $v_{3,k}$ denote the $k$-th vertex in topological ordering in $G_3$ for  $1 \leq k \leq |V_3|$.

Let
\begin{equation*}
\hat{\mathsf{S}}_{\mathrm{IC}}(\hat{v}_{1,i}, \hat{v}_{2,j}, v_{3,k}) = 
\left\{
z ~\middle|~
\begin{array}{l}
  \exists q \in \lf_3(\MP(v_{3,k})) \mbox{ such that } q \in \Subseq(z) \\
  \mbox{and } z \in \Subseq(\hat{\lf}_1(\P(\hat{v}_{1,i}))) \cap \Subseq(\hat{\lf}_2(\P(\hat{v}_{2,j})))
\end{array}
\right\}.
\end{equation*}
Let $\hat{D}_{i,j,k}$ denote the length of a longest string in $\hat{\mathsf{S}}_{\mathrm{IC}}(\hat{v}_{1,i}, \hat{v}_{2,j}, v_{3,k})$.
For convenience, we let $\hat{D}_{i,j,k} = -\infty$ if $\hat{\mathsf{S}}_{\mathrm{IC}}(\hat{v}_{1,i}, \hat{v}_{2,j}, v_{3,k}) = \emptyset$.
The solution to Problem~\ref{prob:SEQ-I-L-nl} (the SEQ-IC-LCS length) is the maximum value of $\hat{D}_{i,j,k}$ for which $v_{3,k}$ has no out-going edges (i.e. $v_{3,k}$ is the end of a maximal path in $G_3$).

$\hat{D}_{i,j,k}$ can be computed as follows:
\begin{enumerate}
\item If both $\hat{v}_{1, i}$ and $\hat{v}_{2, j}$ are cyclic vertices (i.e. $|[\hat{v}_{1, i}]| \geq 2$ and $|[\hat{v}_{2, j}]| \geq 2$),
  \sinote*{added}{%
    then remark that both $\hat{v}_{1, i}$ and $\hat{v}_{2, j}$ have some self-loop(s).
  }%
  There are four cases to consider:
  \begin{enumerate}
    \item If $k = 0$, there are two cases to consider: \label{item:0_1}
    \begin{enumerate}
      \item If $\hat{\lf}_1(\hat{v}_{1, i}) \cap \hat{\lf}_2(\hat{v}_{2, j}) \neq \emptyset$, then clearly $\hat{D}_{i, j, k} = \infty$.
      \item Otherwise, there are two cases to consider: \label{item:ii}
      \begin{enumerate}
        \item If the in-coming edges of $\hat{v}_{1, i}$ are $\hat{v}_{2, j}$ only self-loops,
             then clearly $\hat{D}_{i, j, k}=0$.
        \item Otherwise (\sinote*{added}{$\hat{v}_{1,i}$ has some in-coming edge(s) other than self-loops, or $\hat{v}_{2,j}$ has some in-coming edge(s) other than self-loops}),
              let $\hat{v}_{1,x}$ and $\hat{v}_{2, y}$ be any nodes such that
              $(\hat{v}_{1,x}, \hat{v}_{1, i}) \in \hat{E}_1$ and $(\hat{v}_{2, y}, \hat{v}_{2, j}) \in \hat{E}_2$, 
              respectively.
              Let $s$ be a longest string in the set
              $\Subseq(\hat{\lf}_1(\LMP(\hat{v}_{1,i}))) \cap \Subseq(\hat{\lf}_2(\LMP(\hat{v}_{2,j})))$.
              Assume on the contrary that there is a string
              $t \in \Subseq(\hat{\lf}_1(\LMP(\hat{v}_{1,x}))) \cap \Subseq(\hat{\lf}_2(\LMP(\hat{v}_{2,j})))$
              such that $|t| > |s|$.
              This contradicts that $s$ is a longest common subsequence of 
              $\hat{\lf}_1(\LMP(\hat{v}_{1,i}))$ and $\hat{\lf}_2(\LMP(\hat{v}_{2,j}))$, 
              since  $\Subseq(\hat{\lf}_1(\LMP(\hat{v}_{1,x}))) \cap \Subseq(\hat{\lf}_2(\LMP(\hat{v}_{2,j}))) \subseteq  \Subseq(\hat{\lf}_1(\LMP(\hat{v}_{1,i}))) \cap \Subseq(\hat{\lf}_2(\LMP(\hat{v}_{2,j})))$.
              Hence $|t| \leq |s|$.
              If $\hat{v}_{1, x}$ is a vertex satisfying $\hat{D}_{x,j,k} = |s|$, then $\hat{D}_{i,j,k} = \hat{D}_{x,j,k}$.
              Similarly, if $\hat{v}_{2,y}$ is a vertex satisfying $\hat{D}_{i,y,k} = |s|$,
              then $\hat{D}_{i,j,k} = \hat{D}_{i,y,k}$.
              Note that such $\hat{v}_{1,x}$ or $\hat{v}_{2,y}$ always exists.
      \end{enumerate}
    \end{enumerate}
    \item If $k > 0$ and $\hat{\lf}_1(\hat{v}_{1, i}) \cap \hat{\lf}_2(\hat{v}_{2, j}) \cap \{\lf_3(v_{3, k})\} \neq \emptyset$, there are two cases to consider:
    \begin{enumerate}
      \item If $v_{3, k}$ has no in-coming edges,
          let $\hat{v}_{1,x}$ and $\hat{v}_{2,y}$ be any nodes such that
             $(\hat{v}_{1,x}, \hat{v}_{1, i}) \in \hat{E}_1$ and $(\hat{v}_{2, y}, \hat{v}_{2, j}) \in \hat{E}_2$, 
          respectively \sinote*{added}{(these edges may be self-loops)}.
             If $\hat{D}_{x,y,0} = -\infty$ for all $1 \leq x < i$ and $1 \leq y < j$,
             then clearly $\hat{D}_{i,j,k} = -\infty$. 
             Otherwise, clearly $\hat{D}_{i,j,k} = \infty$.
      \item Otherwise ($v_{3, k}$ has some in-coming edge(s)),
             let 
             $\hat{v}_{1,x}$, $\hat{v}_{2,y}$ and $v_{3,z}$ be any nodes such that
             $(\hat{v}_{1,x}, \hat{v}_{1, i}) \in \hat{E}_1$, $(\hat{v}_{2, y}, \hat{v}_{2, j}) \in \hat{E}_2$ and $(v_{3, z}, v_{3, k}) \in E_3$, 
             respectively \sinote*{added}{(the first two edges may be self-loops)}.
             If $\hat{D}_{x,y,z} = -\infty$ for all $1 \leq x < i$ and $1 \leq y < j$,
             then clearly $\hat{D}_{i,j,k} = -\infty$. 
             Otherwise, $\hat{D}_{i,j,k} = \infty$. 
    \end{enumerate}
    \item If $k > 0$ and $\hat{\lf}_1(\hat{v}_{1, i}) \cap \hat{\lf}_2(\hat{v}_{2, j}) \cap \{\lf_3(v_{3, k})\} = \emptyset$ and $\hat{\lf}_1(\hat{v}_{1, i}) \cap \hat{\lf}_2(\hat{v}_{2, j}) \neq \emptyset$, there are two cases to consider:
    \begin{enumerate}
      \item If \sinote*{modified}{%
            the in-coming edges of $\hat{v}_{1, i}$ are $\hat{v}_{2, j}$ only self-loops,
            }%
            then clearly $\hat{D}_{i, j, k} = -\infty$. 
      \item Otherwise (\sinote*{modified}{$\hat{v}_{1, i}$ has some in-coming edge(s) other than self-loops, or $\hat{v}_{2, j}$ has some in-coming edge(s) other than self-loops}),
            let $\hat{v}_{1,x}$ and $\hat{v}_{2,y}$ be any nodes such that
            $(\hat{v}_{1,x}, \hat{v}_{1, i}) \in \hat{E}_1$ and $(\hat{v}_{2, y}, \hat{v}_{2, j}) \in \hat{E}_2$, 
            respectively.
            If all $\hat{D}_{x,y,k} = -\infty$,
            then clearly $\hat{D}_{i,j,k} = -\infty$.
            Otherwise, clearly $\hat{D}_{i,j,k} = \infty$. 
    \end{enumerate}
    \item If $k > 0$ and $\hat{\lf}_1(\hat{v}_{1, i}) \cap \hat{\lf}_2(\hat{v}_{2, j}) = \emptyset$, there are two cases to consider: \label{item:d}
    \begin{enumerate}
      \item If \sinote*{modified}{the in-coming edges of $\hat{v}_{1,i}$ and $\hat{v}_{2,j}$ are only self-loops}, 
           then clearly $\hat{D}_{i, j, k} = -\infty$.
      \item Otherwise (\sinote*{modified}{$\hat{v}_{1, i}$ has some in-coming edge(s) other than self-loops, or $\hat{v}_{2, j}$ has some in-coming edge(s) other than self-loops}),
            let $\hat{v}_{1,x}$ and $\hat{v}_{2, y}$ be any nodes such that
            $(\hat{v}_{1,x}, \hat{v}_{1, i}) \in \hat{E}_1$ and $(\hat{v}_{2, y}, \hat{v}_{2, j}) \in \hat{E}_2$, 
            respectively.
            Let $s$ be a longest string in $\hat{\mathsf{S}}_{\mathrm{IC}}(\hat{v}_{1,i}, \hat{v}_{2,j}, v_{3,k})$.
            Assume on the contrary that there exists a string 
            $t \in \hat{\mathsf{S}}_{\mathrm{IC}}(\hat{v}_{1,i}, \hat{v}_{2,j}, v_{3,k})$
            such that $|t| > |s|$.
            This contradicts that $s$ is a SEQ-IC-LCS of 
            $\hat{\lf}_1(\LMP(\hat{v}_{1,i}))$, $\hat{\lf}_2(\LMP(\hat{v}_{2,j}))$ and $\lf_3(\MP(v_{3,k}))$, 
            since  $\hat{\mathsf{S}}_{\mathrm{IC}}(\hat{v}_{1,x}, \hat{v}_{2,y}, v_{3,k}) \subseteq  \hat{\mathsf{S}}_{\mathrm{IC}}(\hat{v}_{1,i}, \hat{v}_{2,j}, v_{3,k})$.
            Hence $|t| \leq |s|$.
            If $\hat{v}_{1, x}$ is a vertex satisfying $\hat{D}_{x, j, k} = |z|$, then $\hat{D}_{i,j,k} = \hat{D}_{x,j,k}$.
            Similarly, if $\hat{v}_{2,y}$ is a vertex satisfying $\hat{D}_{i,y,k} = |s|$,
            then $\hat{D}_{i,j,k} = \hat{D}_{i,y,k} $.
            Note that such $\hat{v}_{1,x}$ or $\hat{v}_{2,y}$ always exists.
    \end{enumerate}
  \end{enumerate}
  \item Otherwise (\sinote*{added}{$v_{1,i}$ is not a cyclic vertex and/or $v_{2,j}$ is not a cyclic vertex}), there are four cases to consider:
  \begin{enumerate}
    \item If $k = 0$, there are two cases to consider: \label{item:0_2}
    \begin{enumerate}
      \item If $\hat{\lf}_1(\hat{v}_{1, i}) \cap \hat{\lf}_2(\hat{v}_{2, j}) \neq \emptyset$, there are two cases to consider:
      \begin{enumerate}
        \item If \sinote*{modified}{$\hat{v}_{1, i}$ does not have in-coming edges at all or $\hat{v}_{2, j}$ does not have in-coming edges at all}, 
              then clearly $\hat{D}_{i,j,k} = 1$.
        \item Otherwise (both $\hat{v}_{1, i}$ and $\hat{v}_{2, j}$ have \sinote*{modified}{some in-coming edge(s) including self-loops}), 
              let $\hat{v}_{1,x}$ and $\hat{v}_{2, y}$ be any nodes such that
              $(\hat{v}_{1,x}, \hat{v}_{1, i}) \in \hat{E}_1$ and $(\hat{v}_{2, y}, \hat{v}_{2, j}) \in \hat{E}_2$, 
              respectively.
              Let $s$ be a longest string in the set
              $\Subseq(\hat{\lf}_1$ $(\LMP(\hat{v}_{1,i}))) \cap \Subseq(\hat{\lf}_2(\LMP(\hat{v}_{2,j})))$.
              Assume on the contrary that 
              there is a string $t \in \Subseq(\hat{\lf}_1(\LMP(\hat{v}_{1,x}))) \cap 
              \Subseq(\hat{\lf}_2(\LMP(\hat{v}_{2,y})))$ such that $|t| > |s| - 1$.
              This contradicts that $s$ is a longest common subsequence of 
              $\hat{\lf}_1(\LMP(\hat{v}_{1,i}))$ and $\hat{\lf}_2(\LMP(\hat{v}_{2,j}))$,
              since $\hat{\lf}_1(\hat{v}_{1, i}) \cap \hat{\lf}_2(\hat{v}_{2, j}) \neq \emptyset$.
              Hence $|t| \leq |s| - 1$. 
              If $\hat{v}_{1, x}$ and $\hat{v}_{2, y}$ are vertices
              satisfying $\hat{D}_{x, y, k} = |s| - 1$, then $\hat{D}_{i, j, k} = \hat{D}_{x, y, k} + 1$.
              Note that such $\hat{v}_{1, x}$ and $\hat{v}_{2, y}$ always exist.
       \end{enumerate}
      \item Otherwise, then this case is the same as Case~\ref{item:ii}.
    \end{enumerate}
    \item If $\hat{\lf}_1(\hat{v}_{1, i}) \cap \hat{\lf}_2(\hat{v}_{2, j}) \cap \{\lf_3(v_{3, k})\} \neq \emptyset$, there are three cases to consider:
    \begin{enumerate}
      \item If \sinote*{modified}{$\hat{v}_{1, i}$ does not have in-coming edges at all or $\hat{v}_{2, j}$ does not have in-coming edges at all, and if $v_{3, k}$ does not have in-coming edges}, 
             then clearly $\hat{D}_{i, j, k} = 1$.
      \item If \sinote*{modified}{$\hat{v}_{1, i}$ does not have in-coming edges at all or $\hat{v}_{2, j}$ does not have in-coming edge at all, and if $v_{3, k}$ has some in-coming edge(s)},
             then clearly $\hat{D}_{i, j, k} = -\infty$. 
      \item If both $\hat{v}_{1, i}$ and $\hat{v}_{2, j}$ have \sinote*{modified}{some in-coming edge(s) including self-loops} and $v_{3, k}$ does not have in-coming edges,
             let $\hat{v}_{1,x}$ and $\hat{v}_{2,y}$ be any nodes such that
             $(\hat{v}_{1,x}, \hat{v}_{1, i}) \in \hat{E}_1$ and $(\hat{v}_{2, y}, \hat{v}_{2, j}) \in \hat{E}_2$, 
             respectively.
             Let $s$ be a longest string in the set
             $\Subseq(\hat{\lf}_1(\LMP(\hat{v}_{1,i}))) \cap \Subseq(\hat{\lf}_2(\LMP(\hat{v}_{2,j})))$.
             Assume on the contrary that 
             there exists a string $t \in \Subseq(\hat{\lf}_1(\LMP(\hat{v}_{1,x}))) \cap 
             \Subseq(\hat{\lf}_2(\LMP(\hat{v}_{2,y})))$ such that $|t| > |s| - 1$.
             This contradicts that $s$ is a longest common subsequence of 
             $\hat{\lf}_1(\LMP(\hat{v}_{1,i}))$ and $\hat{\lf}_2(\LMP(\hat{v}_{2,j}))$,
             since $\hat{\lf}_1(\hat{v}_{1, i}) \cap \hat{\lf}_2(\hat{v}_{2, j}) \neq \emptyset$.
             Hence $|t| \leq |s| - 1$. 
             If $\hat{v}_{1, x}$ and $\hat{v}_{2, y}$ are vertices
             satisfying $\hat{D}_{x, y, 0} = |s| - 1$, then $\hat{D}_{i, j, k} = \hat{D}_{x, y, 0} + 1$.
             Note that such $\hat{v}_{1, x}$ and $\hat{v}_{2, y}$ always exist.
      \item Otherwise (all $\hat{v}_{1,i}$, $\hat{v}_{2,j}$, and $\hat{v}_{3,k}$ have \sinote*{modified}{some in-coming edge(s) including self-loops}),
             let $\hat{v}_{1,x}$, $\hat{v}_{2,y}$ and $v_{3,z}$ be any nodes such that
             $(\hat{v}_{1,x}, \hat{v}_{1, i}) \in \hat{E}_1$, $(\hat{v}_{2, y}, \hat{v}_{2, j}) \in \hat{E}_2$, and $(v_{3, z}, v_{3, k}) \in E_3$, 
             respectively.
             Let $s$ be a longest string in $\hat{\mathsf{S}}_{\mathrm{IC}}(\hat{v}_{1,i}, \hat{v}_{2,j}, v_{3,k})$.
             Assume on the contrary that 
             there exists a string $t \in \hat{\mathsf{S}}_{\mathrm{IC}}(\hat{v}_{1,x}, \hat{v}_{2,y}, v_{3,z})$ such that $|t| > |s|-1$.
             This contradicts that $s$ is a SEQ-IC-LCS of 
             $\hat{\lf}_1(\LMP(\hat{v}_{1,i}))$, $\hat{\lf}_2(\LMP(\hat{v}_{2,j}))$ and $\lf_3(\MP(v_{3,k}))$,
             since $\hat{\lf}_1(\hat{v}_{1, i}) \cap \hat{\lf}_2(\hat{v}_{2, j}) \cap {\lf_3(v_{3, k})} \neq \emptyset$.
             Hence $|t| \leq |s| - 1$. 
             If $\hat{v}_{1, x}$, $\hat{v}_{2, y}$ and $v_{3, z}$ are vertices
             satisfying $\hat{D}_{x, y, z} = |s| - 1$, then $\hat{D}_{i, j, k} = \hat{D}_{x, y, z} + 1$.
             Note that such $\hat{v}_{1, x}$, $\hat{v}_{2, y}$ and $v_{3, z}$ always exist.
    \end{enumerate}
    \item If $\hat{\lf}_1(\hat{v}_{1, i}) \cap \hat{\lf}_2(\hat{v}_{2, j}) \cap \{\lf_3(v_{3, k})\} = \emptyset$ and $\hat{\lf}_1(\hat{v}_{1, i}) \cap \hat{\lf}_2(\hat{v}_{2, j}) \neq \emptyset$, 
    there are two cases to consider:
    \begin{enumerate}
      \item \sinote*{modified}{If $\hat{v}_{1, i}$ does not have in-coming edges at all or $\hat{v}_{2, j}$ does not have in-coming edges at all},
            then clearly $\hat{D}_{i, j, k} = -\infty$.
      \item Otherwise (both $\hat{v}_{1, i}$ and $\hat{v}_{2, j}$ have some \sinote*{modified}{in-coming edges including self-loops}),
            let $\hat{v}_{1,x}$ and $\hat{v}_{2,y}$ be any nodes such that
            $(\hat{v}_{1,x}, \hat{v}_{1, i}) \in \hat{E}_1$ and $(\hat{v}_{2, y}, \hat{v}_{2, j}) \in \hat{E}_2$, 
            respectively.
            Let $s$ be a longest string in $\hat{\mathsf{S}}_{\mathrm{IC}}(\hat{v}_{1,i}, \hat{v}_{2,j}, v_{3,k})$.
            Assume on the contrary that 
            there exists a string $t \in \hat{\mathsf{S}}_{\mathrm{IC}}(\hat{v}_{1,x}, \hat{v}_{2,y}, v_{3,k})$ such that $|t| > |s|-1$.
            This contradicts that $s$ is a SEQ-IC-LCS of 
            $\hat{\lf}_1(\LMP(\hat{v}_{1,i}))$, $\hat{\lf}_2(\LMP(\hat{v}_{2,j}))$ and $\lf_3(\MP(v_{3,k}))$,
            since $\hat{\lf}_1(\hat{v}_{1, i}] \cap \hat{\lf}_2(\hat{v}_{2, j}) \neq \emptyset$.
            Hence $|t| \leq |s| - 1$. 
            If $\hat{v}_{1, x}$, $\hat{v}_{2, y}$ and $v_{3, k}$ are vertices
            satisfying $\hat{D}_{x, y, k} = |s| - 1$, then $\hat{D}_{i, j, k} = \hat{D}_{x, y, k} + 1$.
            Note that such $\hat{v}_{1, x}$, $\hat{v}_{2, y}$ and $v_{3, k}$ always exist.
    \end{enumerate}
    \item If $\hat{\lf}_1(\hat{v}_{1, i}) \cap \hat{\lf}_2(\hat{v}_{2, j}) = \emptyset$, then this case is the same as Case~\ref{item:d}.
  \end{enumerate}
\end{enumerate}

The above arguments lead us to the following recurrence:
\begin{eqnarray} 
  \lefteqn{\hat{D}_{i,j,k} =} \nonumber\\ 
   & & \begin{cases}
     \delta \! + \! \max\left(
     \left\{
     \hat{D}_{x, y, k} ~\middle|~ 
     \begin{array}{l}
         (\hat{v}_{1,x},\hat{v}_{1,i}) \! \in \! \hat{E}_1, \\
         (\hat{v}_{2,y},\hat{v}_{2,j}) \! \in \! \hat{E}_2 
     \end{array} \right\} \cup \! \{0\} \right) &
     \begin{array}{l}
       \mbox{if $k=0$ and} \\
       \hat{\lf}_1(\hat{v}_{1, i}) \cap \hat{\lf}_2(\hat{v}_{2, j}) \neq \emptyset;
     \end{array} \\

     \max\!\left( \!
    \begin{array}{l}
      \{ \hat{D}_{x,j,k} \mid {(\hat{v}_{1,x},\hat{v}_{1,i})\! \in \! \hat{E}_1}\}\cup{} \\
      \{ \hat{D}_{i,y,k} \! \mid \! {(\hat{v}_{2,y},\hat{v}_{2,j})\! \in \! \hat{E}_2}\} 
      \cup \{0\} \\
    \end{array} \! \right) &
    \begin{array}{l}
      \mbox{if $k=0$ and} \\
      \hat{\lf}_1(\hat{v}_{1, i}) \cap \hat{\lf}_2(\hat{v}_{2, j}) = \emptyset;
    \end{array}\\
    
    \delta \! + \! \max\!\left( \!
    \left\{ \hat{D}_{x,y,z} \! ~\middle|~ \!
    \begin{array}{l}
      (\hat{v}_{1,x},\hat{v}_{1,i}) \! \in \! \hat{E}_1, \\
      (\hat{v}_{2,y},\hat{v}_{2,j}) \! \in \! \hat{E}_2, \\
      (v_{3,z},v_{3,k}) \! \in \! E_3  \\
      \mbox{or } z = 0
    \end{array} \right\} \cup \{\gamma\} \! \right) &
    \begin{array}{l}
      \mbox{if $k > 0$ and} \\
      \hat{\lf}_1(\hat{v}_{1, i}) \cap \hat{\lf}_2(\hat{v}_{2, j}) \cap \{\lf_3([v_{3, k}])\} \\
      \neq \emptyset;
    \end{array} \\

    \max\!\left( \!
    \left\{ \! \delta+\hat{D}_{x,y,k} \! ~\middle|~ \!
    \begin{array}{l}
      (\hat{v}_{1,x},\hat{v}_{1,i}) \! \in \! \hat{E}_1, \\
      (\hat{v}_{2,y},\hat{v}_{2,j}) \! \in \! \hat{E}_2
    \end{array} \! \right\}
    \cup \{-\infty\} \! \right) &
    \begin{array}{l}
      \mbox{if $k > 0$,}\\
      \hat{\lf}_1(\hat{v}_{1, i}) \cap \hat{\lf}_2(\hat{v}_{2, j}) \cap \{\lf_3(v_{3, k})\} \\
      = \emptyset, \mbox{ and } 
      \hat{\lf}_1(\hat{v}_{1, i}) \! \cap \! \hat{\lf}_2(\hat{v}_{2, j}) \! \neq \! \emptyset;
    \end{array} \\

    \max\!\left( \!
    \begin{array}{l}
      \{ \hat{D}_{x,j,k} \mid {(\hat{v}_{1,x},\hat{v}_{1,i})\! \in \! \hat{E}_1}\} \cup{} \\
      \{ \hat{D}_{i,y,k} \! \mid \! {(\hat{v}_{2,y},\hat{v}_{2,j})\! \in \! \hat{E}_2}\} \cup \{-\infty\}
    \end{array} \! \right)
    & \mbox{otherwise},
       \end{cases}
  \nonumber
  \label{eqn:SEQ-I-L-nl-c-1}
\end{eqnarray}
%
where
\begin{eqnarray*}
\delta & = &
\begin{cases}
  \infty &
    \mbox{if both $\hat{\lf}_1(\hat{v}_{1, i})$ and $\hat{\lf}_2(\hat{v}_{2, j})$ are cyclic vertices;} \\
  1 & \mbox{otherwise},
\end{cases} \\
\gamma & = &
\begin{cases}
  0 &
  \begin{array}{l}
    \mbox{if $\hat{v}_{1, i}$ does not have in-coming edges at all or $\hat{v}_{2, j}$ does not have} \\
    \mbox{in-coming edges at all, and $v_{3, k}$ does not have in-coming edges};
  \end{array}  \\
  -\infty & \mbox{otherwise}. 
\end{cases}
\end{eqnarray*}
In the above recurrence, we use a convention that $\infty + (-\infty) = -\infty$.


We perform preprocessing which transforms $G_1$ and $G_2$ into $\hat{G}_1$ and $\hat{G}_2$ in
$O(|E_1|+|E_2|)$ time with $O(|V_1|+|V_2|)$ space, based on strongly connected components.

To examine the conditions in the above recurrence,
we explicitly construct the intersection of the character labels of the given vertices
$\hat{v}_{1,i} \in \hat{V_1}$, $\hat{v}_{2,j} \in \hat{V_2}$,
and $\hat{v}_{3,k} \in V_3$ by using balanced trees, as follows:
\begin{itemize}
\item Checking whether $\hat{\lf}_1(\hat{v}_{1,i}) \cap \hat{\lf}_2(\hat{v}_{2,j}) = \emptyset$ or $\neq \emptyset$:
Let $\Sigma_1$ and $\Sigma_2$ be the sets of characters that appear in $G_1$ and $G_2$,
respectively.
For every node $\hat{v}_{1,i} \in \hat{V}_1$ of the transformed graph $\hat{G}_1$,
we build a balanced tree $\mathcal{T}_{i}$ which consists of the characters in $\hat{\lf}_1(\hat{v}_i)$.
Since the total number of characters in the original graph $G_1 = (V_1, E_1)$ is equal to $|V_1|$,
we can build the balanced trees $\mathcal{T}_{i}$ for all $i$ in a total of $O(|V_1|\log{|\Sigma_1|})$ time and $O(|V_1|)$ space.
Then, for each fixed $\hat{\lf}_1(\hat{v}_{1,i}) \in \hat{V}_1$,
by using its balanced tree,
the intersection $\hat{\lf}_1(\hat{v}_{1,i}) \cap \hat{\lf}_2(\hat{v}_{2,j})$ can be computed in $O(|V_2|\log|\Sigma_1|)$ time for all $\hat{\lf}_2(\hat{v}_{2,j}) \in V_2$.
Therefore, $\hat{\lf}_1(\hat{v}_{1,i}) \cap \hat{\lf}_2(\hat{v}_{2,j})$ for all $1 \leq i \leq |\hat{V}_1|$ and $1 \leq j \leq |\hat{V}_2|$ can be computed in $O(|V_1||V_2|\log|\Sigma_1|)$ total time.

\item Checking whether $\hat{\lf}_1(\hat{v}_{1,i}) \cap \hat{\lf}_2(\hat{v}_{2,j}) \cap \lf_3(v_{3,k}) = \emptyset$ or $\neq \emptyset$:
While computing $\Sigma_{i,j} = \hat{\lf}_1(\hat{v}_{1,i}) \cap \hat{\lf}_2(\hat{v}_{2,j})$ in the above,
we also build another balanced tree $\mathcal{T}_{i,j}$ which consists of the characters in $\Sigma_{i,j}$ for every $1 \leq i \leq |\hat{V}_1|$ and $1 \leq j \leq |\hat{V}_2|$.
This can be done in $O(|V_1||V_2|\log|\Sigma_1|)$ total time and $O(|V_1||V_2|)$ space.
Then, for each fixed $1 \leq i \leq |\hat{V}_1|$ and $1 \leq j \leq |\hat{V}_2|$,
$\hat{\lf}_1(\hat{v}_{1,i}) \cap \hat{\lf}_2(\hat{v}_{2,j}) \cap \lf_3(v_{3,k})$  can be computed in a total of $O(|V_3|\log|\Sigma_{i,j}|)$ time.
Therefore, $\hat{\lf}_1(\hat{v}_{1,i}) \cap \hat{\lf}_2(\hat{v}_{2,j}) \cap \lf_3(v_{3,k})$ for all $1 \leq i \leq |\hat{V}_1|$, $1 \leq j \leq |\hat{V}_2|$ and,  $1 \leq k \leq |V_3|$ can be computed in $O(|V_1||V_2||V_3|\log|\Sigma|)$ time.
\end{itemize}

Assuming that the above preprocessing for the conditions in the recurrence are all done,
we can compute $\hat{D}_{i,j,k}$ for all $1 \leq i \leq |\hat{V}_1|$,
$1 \leq j \leq |\hat{V}_2|$ and $1 \leq k \leq |V_3|$ using dynamic programming
of size $O(|\hat{V}_1||\hat{V}_2||V_3|)$ in $O(|\hat{E}_1||\hat{E}_2||E_3|)$ time,
in a similar way to the acyclic case for Theorem~\ref{theo:SEQ-I-L-nl}.

Overall, the total time complexity is 
$O(|E_1|+|E_2|+|E_3| + |\hat{V}_1||\hat{V}_2|\log|\Sigma_1| + |\hat{V}_1||\hat{V}_2||V_3|\log|\Sigma| + |\hat{E}_1||\hat{E}_2||E_3|) \subseteq O(|E_1||E_2||E_3|+|V_1||V_2||V_3|\log|\Sigma|)$.

The total space complexity is $O(|V_1||V_2|+|\hat{V}_1||\hat{V}_2||V_3|) \subseteq O(|V_1||V_2||V_3|)$.
\qed
\end{proof}

An example of computing $\hat{D}_{i,j,k}$ using dynamic programming is shown in Figure~\ref{fig:SEQ-I-L-c}.

Algorithm~\ref{algo:SEQ-I-L-nl-c} in Appendix~\ref{sec:SEQ-IC-LCS_cyclic_pseudocode} shows 
a pseudo-code of our algorithm which solves Problem~\ref{prob:SEQ-I-L-nl} 
in the case where $G_1$ and $G_2$ can be cyclic and $G_3$ is acyclic.

\section{Conclusions and Open Questions}\label{chap:conclusion}

In this paper, we introduced the new problem of computing the SEQ-IC-LCS on labeled graphs.
We showed that when the all the input labeled graphs are acyclic, the problem can be solved 
in $O(|E_1||E_2|||E_3)$ time and $O(|V_1||V_2||V_3|)$ space 
by a dynamic programming approach. 
Furthermore, we extend our algorithm to a more general case where
the two target labeled graphs can contain cycles, and presented
an efficient algorithm that runs in 
$O(|E_1||E_2||E_3| + |V_1||V_2||V_3| \log{|\Sigma|})$ time and $O(|V_1||V_2||V_3|)$ space.

Interesting open questions are whether one can extend the framework of our methods
to the other variants STR-IC/EC-LCS and SEQ-EC-LCS of the constrained LCS problems
in the case of labeled graph inputs.
We believe that SEQ-EC-LCS for labeled graphs can be solved by similar methods
to our SEQ-IC-LCS methods, within the same bounds.

\section*{Acknowledgments}
This work was supported by JSPS KAKENHI Grant Numbers
JP21K17705~(YN), JP23H04386~(YN), JP22H03551~(SI), and JP23K18466~(SI).

\bibliographystyle{abbrv}
\bibliography{ref}

\begin{thebibliography}{10}

\bibitem{AbboudBW15}
A.~Abboud, A.~Backurs, and V.~V. Williams.
\newblock Tight hardness results for {LCS} and other sequence similarity
  measures.
\newblock In {\em {FOCS} 2015}, pages 59--78, 2015.

\bibitem{amir}
A.~Amir, M.~Lewenstein, and N.~Lewenstein.
\newblock Pattern matching in hypertext.
\newblock In {\em WADS 1997}, volume 1272 of {\em LNCS}, pages 160--173, 1997.

\bibitem{AnglesABHRV17}
R.~Angles, M.~Arenas, P.~Barcel{\'{o}}, A.~Hogan, J.~L. Reutter, and D.~Vrgoc.
\newblock Foundations of modern query languages for graph databases.
\newblock {\em {ACM} Comput. Surv.}, 50(5):68:1--68:40, 2017.

\bibitem{AoyamaNIIBT18}
K.~Aoyama, Y.~Nakashima, T.~I, S.~Inenaga, H.~Bannai, and M.~Takeda.
\newblock Faster online elastic degenerate string matching.
\newblock In {\em {CPM} 2018}, volume 105 of {\em LIPIcs}, pages 9:1--9:10,
  2018.

\bibitem{ARSLAN}
A.~N. Arslan.
\newblock Regular expression constrained sequence alignment.
\newblock {\em Journal of Discrete Algorithms}, 5(4):647--661, 2007.
\newblock Selected papers from Combinatorial Pattern Matching 2005.

\bibitem{ArslanE05}
A.~N. Arslan and {\"{O}}.~Egecioglu.
\newblock Algorithms for the constrained longest common subsequence problems.
\newblock {\em Int. J. Found. Comput. Sci.}, 16(6):1099--1109, 2005.

\bibitem{BernardiniGPPR22}
G.~Bernardini, P.~Gawrychowski, N.~Pisanti, S.~P. Pissis, and G.~Rosone.
\newblock Elastic-degenerate string matching via fast matrix multiplication.
\newblock {\em {SIAM} J. Comput.}, 51(3):549--576, 2022.

\bibitem{BernardiniPPR20}
G.~Bernardini, N.~Pisanti, S.~P. Pissis, and G.~Rosone.
\newblock Approximate pattern matching on elastic-degenerate text.
\newblock {\em Theor. Comput. Sci.}, 812:109--122, 2020.

\bibitem{BringmannK15}
K.~Bringmann and M.~K{\"{u}}nnemann.
\newblock Quadratic conditional lower bounds for string problems and dynamic
  time warping.
\newblock In {\em {FOCS} 2015}, pages 79--97, 2015.

\bibitem{Caceres23}
M.~C{\'{a}}ceres.
\newblock Parameterized algorithms for string matching to {DAGs}: Funnels and
  beyond.
\newblock In {\em {CPM} 2023}, volume 259 of {\em LIPIcs}, pages 7:1--7:19,
  2023.

\bibitem{SEQECLCS_Chen_2011}
Y.-C. Chen and K.-M. Chao.
\newblock On the generalized constrained longest common subsequence problems.
\newblock {\em Journal of Combinatorial Optimization}, 21(3):383--392, Apr
  2011.

\bibitem{SEQICLCS_2004}
F.~Y. Chin, A.~D. Santis, A.~L. Ferrara, N.~Ho, and S.~Kim.
\newblock A simple algorithm for the constrained sequence problems.
\newblock {\em Information Processing Letters}, 90(4):175 -- 179, 2004.

\bibitem{hypertext}
J.~Conklin.
\newblock Hypertext: An introduction and survey.
\newblock {\em IEEE Computer}, 20(9):17--41, 1987.

\bibitem{bbw089}
T.~C. P.-G. Consortium.
\newblock {Computational pan-genomics: status, promises and challenges}.
\newblock {\em Briefings in Bioinformatics}, 19(1):118--135, 2016.

\bibitem{STRICLCS_DEOROWICZ_2012}
S.~Deorowicz.
\newblock Quadratic-time algorithm for a string constrained lcs problem.
\newblock {\em Information Processing Letters}, 112(11):423 -- 426, 2012.

\bibitem{EquiGMT19}
M.~Equi, R.~Grossi, V.~M{\"{a}}kinen, and A.~I. Tomescu.
\newblock On the complexity of string matching for graphs.
\newblock In {\em {ICALP} 2019}, volume 132 of {\em LIPIcs}, pages 55:1--55:15,
  2019.

\bibitem{EquiMT21}
M.~Equi, V.~M{\"{a}}kinen, and A.~I. Tomescu.
\newblock Graphs cannot be indexed in polynomial time for sub-quadratic time
  string matching, unless {SETH} fails.
\newblock In {\em {SOFSEM} 2021}, volume 12607 of {\em Lecture Notes in
  Computer Science}, pages 608--622, 2021.

\bibitem{GrossiILPPRRVV17}
R.~Grossi, C.~S. Iliopoulos, C.~Liu, N.~Pisanti, S.~P. Pissis, A.~Retha,
  G.~Rosone, F.~Vayani, and L.~Versari.
\newblock On-line pattern matching on similar texts.
\newblock In {\em {CPM} 2017}, volume~78 of {\em LIPIcs}, pages 9:1--9:14,
  2017.

\bibitem{IliopoulosKP21}
C.~S. Iliopoulos, R.~Kundu, and S.~P. Pissis.
\newblock Efficient pattern matching in elastic-degenerate strings.
\newblock {\em Inf. Comput.}, 279:104616, 2021.

\bibitem{RLCSA}
G.~Kucherov, T.~Pinhas, and M.~Ziv-Ukelson.
\newblock Regular language constrained sequence alignment revisited.
\newblock In {\em IWOCA 2011}, pages 404--415, 2011.

\bibitem{Manber}
U.~Manber and S.~Wu.
\newblock Approximate string matching with arbitrary costs for text and
  hypertext.
\newblock In {\em Proc. IAPR}, pages 22--33, 1992.

\bibitem{navarro}
G.~Navarro.
\newblock Improved approximate pattern matching on hypertext.
\newblock {\em Theoretial Computer Science}, 237(1--2):455--463, 2000.

\bibitem{park}
K.~Park and D.~K. Kim.
\newblock String matching in hypertext.
\newblock In {\em Proc. CPM'95}, pages 318--329, 1995.

\bibitem{PSC2011-17}
K.~Shimohira, S.~Inenaga, H.~Bannai, and M.~Takeda.
\newblock Computing longest common substring/subsequence of non-linear texts.
\newblock In {\em {PSC} 2011}, pages 197--208, 2011.

\bibitem{CLCS_Tsai_2003}
Y.-T. Tsai.
\newblock The constrained longest common subsequence problem.
\newblock {\em Information Processing Letters}, 88(4):173 -- 176, 2003.

\bibitem{Wagner_1974_LCS}
R.~A. Wagner and M.~J. Fischer.
\newblock The string-to-string correction problem.
\newblock {\em J. ACM}, 21(1):168--173, Jan. 1974.

\bibitem{STRECLCS_Wang_2013}
L.~Wang, X.~Wang, Y.~Wu, and D.~Zhu.
\newblock A dynamic programming solution to a generalized {LCS} problem.
\newblock {\em Inf. Process. Lett.}, 113(19-21):723--728, 2013.

\bibitem{YonemotoNIB23}
Y.~Yonemoto, Y.~Nakashima, S.~Inenaga, and H.~Bannai.
\newblock Space-efficient {STR-IC-LCS} computation.
\newblock In {\em {SOFSEM} 2023}, volume 13878 of {\em LNCS}, pages 372--384,
  2023.

\end{thebibliography}

\clearpage

\appendix

\section{Pseudo-code for SEQ-IC-LCS of Acyclic Labeled Graphs}
\label{sec:SEQ-IC-LCS_acyclic_pseudocode}

\begin{algorithm2e}[h!]\label{algo:SEQ-I-L-nl}
\caption{Computie the SEQ-IC-LCS length of acyclic labeled graphs }
\renewcommand{\baselinestretch}{0.9}\selectfont
\KwIn{Acyclic labeled graphs $G_1=(V_1,E_1,\lf_1) , G_2=(V_2,E_2,\lf_2), G_3=(V_3,E_3,\lf_3)$}
\KwOut{Length of a longest string in the set \\ $\{z \mid \exists~q \in \lf_3(\MP(G_3)) \mbox{ such that } q \in \Subseq(z) \mbox{ and } z \in \Subseq(G_1) \cap \Subseq(G_2)\}$}
topological sort\ $G_1$, $G_2$, and $G_3$\;
$V_{\mathrm{end}} \leftarrow$ the vertices of $G_3$ which have no out-going edges\;
$n \leftarrow |V_1|$; $m \leftarrow |V_2|$; $l \leftarrow |V_3|+1$\; 
Let $D$ be an $n \times m \times l$ integer array\;

\For{$i \leftarrow 1$ \KwTo $n$}{
 \For{$j \leftarrow 1$ \KwTo $m$}{
  $D_{i,j,0} \leftarrow D_{i,j}$ \tcp*{Algorithm by Shimohira et al.}
 }
}

\For{$i \leftarrow 1$ \KwTo $n$}{
 \For{$j \leftarrow 1$ \KwTo $m$}{
  \For{$k \leftarrow 1$ \KwTo $l$}{
   \If{$\lf_1(v_{1,i})=\lf_2(v_{2,j})=\lf_3(v_{3,k})$}{
    \lIf{there are no in-coming edges to $v_{3,k}$}{ 
      $D_{i,j,k}\leftarrow 1$
    }
    \lElse{
      $D_{i,j,k}\leftarrow -\infty$
    }
    \ForAll{$v_{1,x}$ s.t. $(v_{1,x},v_{1,i})\in E_1$}{
     \ForAll{$v_{2,y}$ s.t. $(v_{2,y},v_{2,j})\in E_2$}{
      \If{there are no in-coming edges to $v_{3,k}$}{
          \If{$D_{i,j,k}<1+D_{x,y,0}$}{
            $D_{i,j,k}\leftarrow 1+D_{x,y,0}$\;
          }
        }
      \Else{
        \ForAll{$v_{3,z}$ s.t. $(v_{3,z},v_{3,k})\in E_3$}{
          \lIf{$D_{i,j,k}<1+D_{x,y,z}$}{
            $D_{i,j,k}\leftarrow 1+D_{x,y,z}$
          }
        }
      }
     }
    }
   }
   \ElseIf{$\lf_1(v_{1,i})=\lf_2(v_{2,j})\neq\lf_3(v_{3,k})$}{
      $D_{i,j,k}\leftarrow -\infty$\;
     \ForAll{$v_{1,x}$ s.t. $(v_{1,x},v_{1,i})\in E_1$}{
      \ForAll{$v_{2,y}$ s.t. $(v_{2,y},v_{2,j})\in E_2$}{
       \lIf{$D_{i,j,k}<1+D_{x,y,k}$}{
        $D_{i,j,k}\leftarrow 1+D_{x,y,k}$
       }
      }
     }
    }
  \Else{
    $D_{i,j,k}\leftarrow -\infty$\;
    \ForAll{$v_{1,x}$ s.t. $(v_{1,x},v_{1,i})\in E_1$}{
     \lIf{$D_{i,j,k}<D_{x,j,k}$}{
      $D_{i,j,k}\leftarrow D_{x,j,k}$
     }
    }
    \ForAll{$v_{2,y}$ s.t. $(v_{2,y},v_{2,j})\in E_2$}{
     \lIf{$D_{i,j,k}<D_{i,y,k}$}{
      $D_{i,j,k}\leftarrow D_{i,y,k}$ 
     }
    }
   }
  }
 }
}
\Return $\max\{D_{i,j,k} \mid 1 \leq i \leq n, 1 \leq j \leq m, k \in V_{\mathrm{end}}\}$\;
\end{algorithm2e}

\clearpage

\section{Pseudo-code for SEQ-IC-LCS of Cyclic Labeled Graphs}
\label{sec:SEQ-IC-LCS_cyclic_pseudocode}

\begin{algorithm2e}[h!]\label{algo:SEQ-I-L-nl-c} 
\caption{Compute the SEQ-IC-LCS length of cyclic labeled graphs}
\renewcommand{\baselinestretch}{0.9}\selectfont
\KwIn{Labeled graphs $G_1=(V_1,E_1,\lf_1)$, $G_2=(V_2,E_2,\lf_2)$, and acyclic labeled graph $G_3=(V_3,E_3,\lf_3)$}
\KwOut{Length of a longest string in the set \\ $\{z \mid \exists~q \in \lf_3(\MP(G_3)) \mbox{ such that } q \in \Subseq(z) \mbox{ and } z \in \Subseq(G_1) \cap \Subseq(G_2)\}$}

$G'_1 \leftarrow $ Graph obtained from $G_1$ based on its strongly connected components\;
$G'_2 \leftarrow $ Graph obtained from $G_2$ based on its strongly connected components\;
topological sort\ $\hat{G}_1$, $\hat{G}_2$, $G_3$\;
$V_{\mathrm{end}} \leftarrow$ the vertices of $G_3$ which have no out-going edges\;
Let $\hat{D}$ be a $|\hat{V}_1| \times |\hat{V}_2| \times (|V_3|+1)$ integer array\;
Let $M$ be a $|\hat{V}_1| \times |\hat{V}_2|$ integer array\;
\For{$i \leftarrow 1$ \KwTo $|\hat{V}_1|$}{
  \For{$j \leftarrow 1$ \KwTo $|\hat{V}_2|$}{
    \lIf{$\hat{\lf}_1(\hat{v}_{1, i}) \cap \hat{\lf}_2(\hat{v}_{2, j}) \neq \emptyset$}{
      $M_{i,j} \leftarrow 1$
    }
    \lElse{
      $M_{i,j} \leftarrow 0$
    }
  }
}
Let $S$ be a $|\hat{V}_1| \times |\hat{V}_2|$ array to store pointers to balanced trees\;
Let $B$ be a $|\hat{V}_1|$ array to store pointers to balanced trees\;

\For{$i \leftarrow 1$ \KwTo $|\hat{V}_1|$}{
  $B_i \leftarrow$ balanced tree of characters in $\hat{\lf}_1(\hat{v}_{1,i})$\;
}
\For{$i \leftarrow 1$ \KwTo $|\hat{V}_1|$}{
  \For{$j \leftarrow 1$ \KwTo $|\hat{V}_2|$}{
    \If{some characters in $\hat{\lf}_2(\hat{v}_{2,j})$ are in $B_i$}{ 
      $M_{i,j} \leftarrow 1$\;
      $S_{i,j} \leftarrow$ balanced trees of all characters in $B_i \cap \hat{\lf}_2(\hat{v}_{2,j})$\;
    }
    \lElse{
      $M_{i,j} \leftarrow 0$
    }
  }
}

Compute $\hat{D}_{i,j,0}$; \tcp*{Algorithm~\ref{algo:k=0}}

 \For{$i \leftarrow 1$ \KwTo $|\hat{V}_1|$}{
  \For{$j \leftarrow 1$ \KwTo $|\hat{V}_2|$}{
    \For{$k \leftarrow 1$ \KwTo $|V_3|$}{
      \If{$(\hat{v}_{1,i},\hat{v}_{1,i}) \in \hat{E}_1$ and $(\hat{v}_{2,j},\hat{v}_{2,j}) \in \hat{E}_2$}{
        Compute $\hat{D}_{i,j,k}$ if $(\hat{v}_{1,i},\hat{v}_{1,i}) \in \hat{E}_1$ and $(\hat{v}_{2,j},\hat{v}_{2,j}) \in \hat{E}_2$ \tcp*{Algorithm~\ref{algo:cycles}}
      }
      \Else{
        Compute $\hat{D}_{i,j,k}$ if $(\hat{v}_{1,i},\hat{v}_{1,i}) \notin \hat{E}_1$ or $(\hat{v}_{2,j},\hat{v}_{2,j}) \notin \hat{E}_2$ \tcp*{Algorithm~\ref{algo:not-cycles}}
      }
    }
  }
 }

 \Return $\max\{\hat{D}_{i,j,k} \mid 1 \leq i \leq |\hat{V}_1|, 1 \leq j \leq |\hat{V}_2|, k \in V_{\mathrm{end}}\}$\;
\end{algorithm2e}

\begin{algorithm2e}[th]\label{algo:k=0}
\caption{Compute $\hat{D}_{i,j,0}$}
\renewcommand{\baselinestretch}{0.9}\selectfont
\For{$i \leftarrow 1$ \KwTo $|\hat{V}_1|$}{
 \For{$j \leftarrow 1$ \KwTo $|\hat{V}_2|$}{
  \If{$(\hat{v}_{1,i},\hat{v}_{1,i}) \in \hat{E}_1$ and $(\hat{v}_{2,j},\hat{v}_{2,j}) \in \hat{E}_2$}{
    \If{$M_{i,j} = 1$}{
      $\hat{D}_{i,j,k}\leftarrow \infty$\;
    }
    \Else{
      $\hat{D}_{i,j,0}\leftarrow 0$\;
      \ForAll{$\hat{v}_{1,x}$ such that $(\hat{v}_{1,x},\hat{v}_{1,i})\in \hat{E}_1$}{
       \If{$\hat{D}_{i,j,0}<\hat{D}_{x,j,0}$}{
         $\hat{D}_{i,j,0}\leftarrow \hat{D}_{x,j,0}$\;
       }
      }
      \ForAll{$\hat{v}_{2,y}$ such that $(\hat{v}_{2,y},\hat{v}_{2,j})\in \hat{E}_2$}{
       \If{$\hat{D}_{i,j,0}<\hat{D}_{i,y,0}$}{
        $\hat{D}_{i,j,0}\leftarrow \hat{D}_{i,y,0}$\; 
       }
      }
    }
  }
  \Else{
    \If{$M_{i,j} = 1$}{
      $\hat{D}_{i,j,k}\leftarrow 1$\;
      \ForAll{$\hat{v}_{1,x}$ such that $(\hat{v}_{1,x},\hat{v}_{1,i})\in \hat{E}_1$}{
        \ForAll{$\hat{v}_{2,y}$ such that $(\hat{v}_{2,y},\hat{v}_{2,j})\in \hat{E}_2$}{
          \If{$\hat{D}_{i,j,0}<1+\hat{D}_{x,y,0}$}{
            $\hat{D}_{i,j,0}\leftarrow 1+\hat{D}_{x,y,0}$\;
          }
        }
      }
    }
    \Else{
      $\hat{D}_{i,j,0}\leftarrow 0$\;
      \ForAll{$\hat{v}_{1,x}$ such that $(\hat{v}_{1,x},\hat{v}_{1,i})\in \hat{E}_1$}{
       \If{$\hat{D}_{i,j,0}<\hat{D}_{x,j,0}$}{
         $\hat{D}_{i,j,0}\leftarrow \hat{D}_{x,j,0}$\;
       }
      }
      \ForAll{$\hat{v}_{2,y}$ such that $(\hat{v}_{2,y},\hat{v}_{2,j})\in \hat{E}_2$}{
       \If{$\hat{D}_{i,j,0}<\hat{D}_{i,y,0}$}{
        $\hat{D}_{i,j,0}\leftarrow \hat{D}_{i,y,0}$\; 
       }
      }
    }
  }
 }
}
\end{algorithm2e}

\begin{algorithm2e}[th]\label{algo:cycles}
  \caption{Compute $\hat{D}_{i,j,k}$ when $(\hat{v}_{1,i},\hat{v}_{1,i}) \in \hat{E}_1$ and $(\hat{v}_{2,j},\hat{v}_{2,j}) \in \hat{E}_2$}
  \renewcommand{\baselinestretch}{0.9}\selectfont
\If{$\hat{\lf}_1(\hat{v}_{1, i}) \cap \hat{\lf}_2(\hat{v}_{2, j}) \cap \{\lf_3(v_{3, k})\} \neq \emptyset$}{ 
            $\hat{D}_{i,j,k}\leftarrow -\infty$\;
            \ForAll{$\hat{v}_{1,x}$ such that $(\hat{v}_{1,x},\hat{v}_{1,i})\in \hat{E}_1$}{
            \ForAll{$\hat{v}_{2,y}$ such that $(\hat{v}_{2,y},\hat{v}_{2,j})\in \hat{E}_2$}{
              \If{there are no in-coming edges to $v_{3,k}$}{
                \If{$\hat{D}_{x,y,0} \neq -\infty$}{
                  $\hat{D}_{i,j,k}\leftarrow \infty$\;
                }
              }
             \Else{
               \ForAll{$\hat{v}_{3,z}$ such that $(\hat{v}_{3,z},\hat{v}_{3,k})\in \hat{E}_3$}{
                \If{$\hat{D}_{x,y,z} \neq -\infty$}{
                  $\hat{D}_{i,j,k}\leftarrow \infty$\;
                 }
               }
             }
            }
           }
          }
\ElseIf{$M_{i,j}=1$}{
          $\hat{D}_{i,j,k}\leftarrow -\infty$\;
          \ForAll{$\hat{v}_{1,x}$ such that $(\hat{v}_{1,x},\hat{v}_{1,i})\in \hat{E}_1$}{
           \ForAll{$\hat{v}_{2,y}$ such that $(\hat{v}_{2,y},\hat{v}_{2,j})\in \hat{E}_2$}{
            \If{$\hat{D}_{x,y,k} \neq -\infty$}{
             $\hat{D}_{i,j,k}\leftarrow \infty$\;
            }
           }
          }
      }
\Else{
         $\hat{D}_{i,j,k}\leftarrow -\infty$\;
         \ForAll{$v_{1,x}$ such that $(v_{1,x},v_{1,i})\in E_1$}{
          \If{$\hat{D}_{i,j,k}<\hat{D}_{x,j,k}$}{
           $\hat{D}_{i,j,k}\leftarrow \hat{D}_{x,j,k}$\;
          }
         }
         \ForAll{$v_{2,y}$ such that $(v_{2,y},v_{2,j})\in E_2$}{
          \If{$\hat{D}_{i,j,k}<\hat{D}_{i,y,k}$}{
           $\hat{D}_{i,j,k}\leftarrow \hat{D}_{i,y,k}$\; 
          }
         }
      }
\end{algorithm2e}

\begin{algorithm2e}[th]\label{algo:not-cycles}
  \caption{Compute $\hat{D}_{i,j,k}$ when $(\hat{v}_{1,i},\hat{v}_{1,i}) \notin \hat{E}_1$ or $(\hat{v}_{2,j},\hat{v}_{2,j}) \notin \hat{E}_2$}
  \renewcommand{\baselinestretch}{0.9}\selectfont
  \If{$\hat{\lf}_1(\hat{v}_{1, i}) \cap \hat{\lf}_2(\hat{v}_{2, j}) \cap \{\lf_3(v_{3, k})\} \neq \emptyset$}{ 
    \If{there are no in-coming edges to $v_{3,k}$}{ 
      $\hat{D}_{i,j,k}\leftarrow 1$\;
    }
    \Else{
      $\hat{D}_{i,j,k}\leftarrow -\infty$\;
    }
    \ForAll{$\hat{v}_{1,x}$ such that $(\hat{v}_{1,x},v_{1,i})\in \hat{E}_1$}{
     \ForAll{$\hat{v}_{2,y}$ such that $(\hat{v}_{2,y},\hat{v}_{2,j})\in \hat{E}_2$}{
      \If{there are no in-coming edges to $v_{3,k}$}{
          \If{$\hat{D}_{i,j,k}<1+\hat{D}_{x,y,0}$}{
            $\hat{D}_{i,j,k}\leftarrow 1+\hat{D}_{x,y,0}$\;
          }
       }
      \Else{ 
        \ForAll{$v_{3,z}$ such that $(v_{3,z},v_{3,k})\in E_3$}{
          \If{$\hat{D}_{i,j,k}<1+\hat{D}_{x,y,z}$}{
            $\hat{D}_{i,j,k}\leftarrow 1+\hat{D}_{x,y,z}$\;
          }
        }
      }
    }
   }
  }
   \ElseIf{$M_{i,j} = 1$}{
      $\hat{D}_{i,j,k}\leftarrow -\infty$\;
     \ForAll{$\hat{v}_{1,x}$ such that $(\hat{v}_{1,x},\hat{v}_{1,i})\in \hat{E}_1$}{
      \ForAll{$\hat{v}_{2,y}$ such that $(\hat{v}_{2,y},\hat{v}_{2,j})\in \hat{E}_2$}{
       \If{$\hat{D}_{i,j,k}<1+\hat{D}_{x,y,k}$}{
        $\hat{D}_{i,j,k}\leftarrow 1+\hat{D}_{x,y,k}$\;
       }
      }
     }
    }
  \Else{
    $\hat{D}_{i,j,k}\leftarrow -\infty$\;
    \ForAll{$\hat{v}_{1,x}$ such that $(\hat{v}_{1,x},\hat{v}_{1,i})\in \hat{E}_1$}{
     \If{$\hat{D}_{i,j,k}<\hat{D}_{x,j,k}$}{
      $\hat{D}_{i,j,k}\leftarrow \hat{D}_{x,j,k}$\;
     }
    }
    \ForAll{$\hat{v}_{2,y}$ such that $(\hat{v}_{2,y},\hat{v}_{2,j})\in \hat{E}_2$}{
     \If{$\hat{D}_{i,j,k}<\hat{D}_{i,y,k}$}{
      $\hat{D}_{i,j,k}\leftarrow \hat{D}_{i,y,k}$\; 
     }
    }
   }
\end{algorithm2e}

\end{document}